\documentclass[manuscript,nonacm]{acmart}

\usepackage{graphicx}
\usepackage{color,xcolor}
\usepackage{colortbl}
\usepackage{multirow}
\usepackage{booktabs}
\usepackage{url}
\usepackage{graphicx}
\usepackage{listings}
\usepackage{color}
\usepackage{multicol}
\usepackage{latexsym}
\usepackage{xspace}
\usepackage{wrapfig}
\usepackage{amsfonts,amsmath}
\usepackage[normalem]{ulem}

\usepackage{xspace}%
\PassOptionsToPackage{hyphens}{url}\usepackage{hyperref}

\definecolor{black}{RGB}{0,0,0}
\definecolor{gray}{RGB}{102,102,102}        
\definecolor{function}{RGB}{0,102,153}      
\definecolor{lightgreen}{RGB}{102,153,0}    
\definecolor{bluegreen}{RGB}{51,153,126}    
\definecolor{magenta}{RGB}{217,74,122}  
\definecolor{orange}{RGB}{226,102,26}       
\definecolor{purple}{RGB}{125,71,147}       
\definecolor{green}{RGB}{113,138,98}        
\usepackage[T1]{fontenc}
\usepackage[scaled]{beramono}
\usepackage{microtype}
\usepackage{enumitem}

\PassOptionsToPackage{hyphens}{url}
\usepackage{hyperref}
\usepackage{url}
\hypersetup{
	colorlinks=true,
	linkcolor=blue,
	citecolor=red,
	urlcolor=magenta,
}

\lstdefinelanguage{parameterized}{
  firstnumber=1,
  numberstyle=\tiny\color{black},
  tabsize=2,
  numbers=left,
  morekeywords = [3]{while,if,then,else,do,done},
  morekeywords = [4]{true,false,and,or},
  morekeywords = [5]{send,recv,break},
  morekeywords = [6]{bool,+,=,:=,.,;,,,-,!,=,~,>,<},
  keywordstyle = [3]\color{bluegreen},
  keywordstyle = [4]\color{lightgreen},
  keywordstyle = [5]\color{magenta},
  keywordstyle = [6]\color{orange},
  sensitive = true,
  morecomment = [l][\color{gray}]{//},
  morecomment = [s][\color{gray}]{/*}{*/},
  morecomment = [s][\color{gray}]{/**}{*/},
  morestring = [b][\color{purple}]",
  morestring = [b][\color{purple}]',
  literate=
  {�}{{\'a}}1 {�}{{\'e}}1 {�}{{\'i}}1 {�}{{\'o}}1 {�}{{\'u}}1
  {�}{{\'A}}1 {�}{{\'E}}1 {�}{{\'I}}1 {�}{{\'O}}1 {�}{{\'U}}1
  {�}{{\`a}}1 {�}{{\`e}}1 {�}{{\`i}}1 {�}{{\`o}}1 {�}{{\`u}}1
  {�}{{\`A}}1 {�}{{\'E}}1 {�}{{\`I}}1 {�}{{\`O}}1 {�}{{\`U}}1
  {�}{{\"a}}1 {�}{{\"e}}1 {�}{{\"i}}1 {�}{{\"o}}1 {�}{{\"u}}1
  {�}{{\"A}}1 {�}{{\"E}}1 {�}{{\"I}}1 {�}{{\"O}}1 {�}{{\"U}}1
  {�}{{\^a}}1 {�}{{\^e}}1 {�}{{\^i}}1 {�}{{\^o}}1 {�}{{\^u}}1
  {�}{{\^A}}1 {�}{{\^E}}1 {�}{{\^I}}1 {�}{{\^O}}1 {�}{{\^U}}1
  {�}{{\~A}}1 {�}{{\~a}}1 {�}{{\~O}}1 {�}{{\~o}}1
  {�}{{\oe}}1 {�}{{\OE}}1 {�}{{\ae}}1 {�}{{\AE}}1 {�}{{\ss}}1
  {?}{{\H{u}}}1 {?}{{\H{U}}}1 {?}{{\H{o}}}1 {?}{{\H{O}}}1
  {�}{{\c c}}1 {�}{{\c C}}1 {�}{{\o}}1 {�}{{\r a}}1 {�}{{\r A}}1
  {�}{{\euro}}1 {�}{{\pounds}}1 {�}{{\guillemotleft}}1
  {�}{{\guillemotright}}1 {�}{{\~n}}1 {�}{{\~N}}1 {�}{{?`}}1
}

\newcommand{\modulo}{{\sf ~mod~}}

\usepackage{algorithmicx}
\usepackage{algpseudocode}
\usepackage{algorithm}
\usepackage{tikz}
\usetikzlibrary{arrows,automata,shapes,cd,shapes.geometric}
\usetikzlibrary{positioning,angles,quotes,calc}

\usepackage{macros}
\usepackage{ioa_code}

\newcommand{\remove}[1]{}

\newcommand{\tlap}[0]{\textsc{TLA}\textsuperscript{+}} 

\newcommand{\myparagraph}[1]{\vspace{1em}{\noindent\bfseries #1{.}}}

\usepackage{lipsum} 

\setcopyright{none}
\copyrightyear{2022}
\acmYear{2022}
\acmDOI{}
\acmConference[]{}{}{}
\acmBooktitle{}
\acmPrice{}
\acmISBN{}

\begin{document}

\title{Holistic Verification of Blockchain Consensus}

%

\author{Nathalie Bertrand}
\email{nathalie.bertrand@inria.fr}
\affiliation{%
  \institution{Univ Rennes, Inria, CNRS, IRISA}
  \city{Rennes}         
  \country{France}   
}

\author{Vincent Gramoli}
\email{vincent.gramoli@sydney.edu.au}
\affiliation{%
  \institution{University of Sydney}
  \city{Sydney}         
  \country{Australia}   
}

\author{Igor Konnov}
\email{igor@informal.systems}
\affiliation{%
  \institution{Informal Systems}
  \city{Wien}         
  \country{Austria}   
}

\author{Marijana Lazi\'{c}}
\email{lazic@in.tum.de}
\affiliation{%
  \institution{TU Munich}
  \city{Munich}         
  \country{Germany}   
}

\author{Pierre Tholoniat}
\email{pierre@cs.columbia.edu}
\affiliation{%
  \institution{Columbia University}
  \city{New York City}         
  \country{USA}   
}

\author{Josef Widder}
\email{josef@informal.systems}
\affiliation{%
  \institution{Informal Systems}
  \city{Wien}         
  \country{Austria}   
}

\begin{abstract}
Blockchain has recently attracted the attention of the industry due,
    in part, to its ability to automate asset transfers.  
    It requires
    distributed participants to reach a consensus on a block despite the
    presence of malicious (a.k.a. Byzantine) participants.  
    Malicious participants exploit regularly weaknesses of these blockchain consensus algorithms, 
    with sometimes devastating consequences.
    In fact, these weaknesses are quite common and are well illustrated by the flaws 
    in the hand-written 
    proofs of existing blockchain consensus protocols~\cite{TG22}.
    Paradoxically,
    until now, no blockchain consensus has been holistically
    verified using model checking.

    In this paper, we remedy this paradox by model checking for the
    first time a blockchain consensus used in industry.  We propose
    a holistic approach to verify the consensus algorithm of the
    Red Belly Blockchain~\cite{CNG21}, for any number $n$ of processes
    and any number $f<n/3$ of Byzantine processes.
    We decompose directly the algorithm pseudocode in two parts---an inner broadcast
    algorithm and an outer decision algorithm---each modelled as a
    threshold automaton~\cite{KLVW17popl}, and we formalize their expected
    properties in linear-time temporal logic. We then automatically
    check the inner broadcasting algorithm, under a carefully identified
    fairness assumption. For the verification of the outer algorithm, we
    simplify the model of the inner algorithm by relying on its checked
    properties. Doing so, 
    we formally verify not only the 
    safety properties of the Red Belly
    Blockchain consensus but also its liveness 
    in about 70~seconds.
\end{abstract}


\maketitle

\section{Introduction}
\subsection{Context}
Today, the market capitalization of the seminal blockchain, called
Bitcoin, is about \$803B\footnote{\url{https://coinmarketcap.com}.},
which incentivizes malicious participants to find problematic
executions that would allow them to steal financial assets.  As the
blockchain requires a distributed set of machines to agree on a unique
block of transactions to be appended to the chain, attackers
naturally try to exploit consensus vulnerabilities: they force
participants to disagree so that they wrongly believe that two
conflicting transactions are legitimate, leading to what is known as a
\emph{double spending}. In 2014, malicious participants managed to
exploit Bitcoin consensus vulnerabilities to steal \$83,000 through a
network attack. In August 2021, 570,000 transactions were reverted in a
more recent version of Bitcoin, Bitcoin SV, by forcing its blockchain
consensus protocol to violate its safety property (i.e.,
agreement). With 3 attacks on the same blockchain within 4
months, thefts are becoming
commonplace.\footnote{\url{cointelegraph.com/news/bitcoin-sv-rocked-by-three-51-attacks-in-as-many-months}}
Unsurprisingly, various bugs in specifications and in proofs of
blockchain consensus protocols appear in the literature~\cite{AGD17,Sut20}.
This is illustrated by the flaw in the consensus algorithms 
now used in in-production blockchains~\cite{TG22}.
The crux of the problem is that reasoning about distributed executions
of blockchain consensus protocols is hard due to several sources of
non-determinism, and in particular asynchrony and faults.
%
As a result, formally verifying that a blockchain consensus protocol
is safe and live is key to mitigate financial losses.

Recent progress in mechanical proofs represent the first steps towards verifying blockchain consensus.
For instance, parameterized model
checking aims at verifying algorithms for an arbitrary number $n$ of
processes~\cite{2015Bloem} that is unknown at design time. In some
contexts, it reduces the model checking for any fault number $f$ and
its upper bound $t$ to bounded model checking
questions~\cite{FismanKL08}. The threshold automaton (TA) framework
for communication-closed
algorithms~\cite{KLVW17popl,Balasubramanian20} targets algorithms with
thresholds in guards such as ``number of messages from distinct
processes exceeds $2t+1$'', and in the resilience condition, typically
of the form $n>3t$. The parameterized model checking of threshold
automata builds upon a reduction~\cite{EF82,Lipton75} that reorders
steps of asynchronous executions to obtain simpler executions, which
are equivalent to the original executions with respect to safety and
liveness properties. 
Such a technique has recently proved instrumental in verifying fully
asynchronous parts of consensus algorithms, like broadcast
algorithms~\cite{KLVW17popl}.

Due to the famous unfeasability of deterministic consensus in
asynchronous setting~\cite{FLP85}, this promising method was not
applied to proving deterministic consensus algorithms\footnote{Here
    deterministic means that randomization is forbidden. However, the
    environment (\emph{e.g.}, communication delays, scheduler)
    introduces non-determinism in the algorithm execution.}.
In fact, the aforementioned reduction technique cannot apply to
partial synchrony~\cite{DLS88}: moving the message reception step to a
later point in the execution might violate an assumed message delay.
Yet, these delays are important as typical partially synchronous
consensus algorithms feature timers to catch
up with the unknown bound on the delay to receive a message. 
Most known verification techniques therefore target either synchronous
(lock-step) or asynchronous semantics. In addition, partially
synchronous consensus algorithms generally rely on a coordinator
process that helps other processes converge and whose identifier
rotates across rounds.  Some efforts have been devoted to proving the
termination of partially synchronous consensus algorithms, like Paxos,
assuming synchrony~\cite{HawblitzelHKLPR15}.  The drawback is that
such algorithms aim at tolerating non-synchronous periods before
reaching a global stabilization time (GST) after which they
terminate. Proving that such an algorithm terminates under synchrony
does not show that the algorithm would also terminate if processes
reached GST at different points of their execution. Instead, one would
also need to show that correct processes can catch up in the same
round. This would, in turn, require proving the correctness of a
synchronizer algorithm~\cite{DLS88}.

Verifying consensus is even more subtle when processes are Byzantine
as they can execute arbitrary steps, changing their local state and
the values they share. One needs to reason about executions with all
possible scenarios resulting from arbitrary behaviors, multiplying the
already large number of interleaved executions.  The verification of
such algorithms is thus either restricted to showing safety properties,
like agreement and validity, and ignoring liveness~\cite{Lam11a}; or 
to proving separate parts of the blockchain consensus~\cite{LD20}.
Such noticeable efforts are well illustrated with the series of attempts to verify 
the Stellar Consensus Protocol (SCP)~\cite{LLM19}: 
Ivy is used to model 
a slightly different algorithm in which the 
key novelty of the SCP consensus algorithm, which is 
its dynamic quorum system~\cite{NW03} called quorum slices, is replaced by static quorums~\cite{LGM19}.
In addition, it relies on axioms some of which 
are proved separately in Isabelle/HOL without being linked to the 
Ivy axioms~\cite{Ste22}.
%
Without a holistic approach, the verification of the components of a protocol does not  
imply that the protocol is verified.
\subsection{Contributions}
In this paper, we verify holistically the safety and liveness properties
of a Byzantine consensus used in the Red Belly Blockchain system~\cite{CNG21}, a scalable blockchain used in production.
Our approach is holistic because it starts from the pseudocode of the distributed algorithm as typically 
presented in the distributed computing literature, models this pseudocode and its components into disambiguated 
threshold automata (\TA{}s), model checks the desired properties of these components expressed in \LTL{} formulae, 
simplifies the \TA of the consensus algorithm with these verified properties and model checks the safety and liveness of the consensus protocol. The advantage is that the formally verified algorithm matches the pseudocode and no user-defined invariants or proofs need to be checked, which drastically reduces the risks of human errors.
\begin{enumerate}
    \item \sloppy{We formally verify a Byzantine consensus 
    	 algorithm~\cite{CGLR18} used for e-voting~\cite{CCC20}, accountability~\cite{CGG19}
	 and blockchain~\cite{CNG21}. 
          This consensus algorithm now runs in the network of the Red Belly Blockchain~\cite{CNG21} maintained 
          by the Red Belly Network company. It
          executes in asynchronous rounds that broadcast binary values
          and compare the delivered values to the parity of the round to decide. 
          To model check the algorithm holistically, we replace the partial synchrony 
          assumption by a fairness assumption. Interestingly, our fairness assumption only
          requires that in any infinite sequence of rounds, there exists a
          round where, at all correct processes, a broadcast instance delivers
          the same binary value, or bit, first.}
    \item We exploit the 
    	  modularity of distributed algorithms in
          parameterized model checking. We first model the consensus algorithm 
          into two simpler algorithms modeled as threshold automata (\TA{}s): (i)~an inner broadcast \TA modeling a binary value variant of the reliable broadcast~\cite{MMR14} and 
          (ii)~an outer decision \TA modeling a round-based execution that inspects 
          the delivered messages~\cite{CGLR18} to decide.
          We express the guarantees of the inner broadcast
          primitive as temporal logic properties that we automatically verify
          and we replace the inner \TA in the global \TA by
          a gadget TA that captures the proven temporal specification. 
          We automatically verify the global \TA with model checking.
    \item We show the practicality of our verification technique
  	 by running the parameterized model checker ByMC~\cite{KLVW17popl}
          for any number $n$
          of processes and any arbitrary number $f<n/3$ of Byzantine processes.  
          We compare the execution times when model checking the naive \TA
          encoding the consensus algorithm and when model checking both the
          inner \TA encoding the broadcast algorithm and then
          the outer \TA. 
          We demonstrate empirically that, although a parallel execution of 
          ByMC on 64 cores could not prove the safety of the naive \TA within 
          3 days, it proves both the liveness and safety of the simplified \TA
          in about 70 seconds.
\end{enumerate}



\subsection{Outline}
In Section~\ref{sec:notations} we introduce our preliminary
definitions, in Section~\ref{sec:da-models} we model our binary
value broadcast algorithm pseudocode into a corresponding threshold automaton, in Section~\ref{sec:composition} we 
explain how the formal verification of the properties of the broadcast algorithm helps us model check the consensus algorithm
and in Section~\ref{sec:verification} we verify the consensus
algorithm.
In Section~\ref{sec:expe} we present the results of the
model checker.
In Section~\ref{sec:rel}, we present the related work and
in Section~\ref{sec:conclusion}, we conclude.  In the Appendix we explain the
multiple-round TA to one-round TA
reduction~(\ref{app:multi2oneround}), provide examples related to
fairness~(\ref{app:no-termination}), missing
proofs~(\ref{app:same-est} and~\ref{app:corollary}) and detailed
specifications~(\ref{app:largeta} and~\ref{sec:ta-spec-termination}).


\section{Preliminaries}\label{sec:notations}

The consensus algorithm runs over $n$ asynchronous sequential
processes from the set $\Pi = \{p_1,\ldots,p_n\}$.
The processes communicate by exchanging messages through
an asynchronous reliable fully connected point-to-point network, hence there is no
bound on the delay to transfer a message but this delay is finite.

\myparagraph{Failure model}
Up to $t<n/3$ processes can
exhibit a {\it Byzantine} behavior~\cite{PSL80}, and behave
arbitrarily.  We refer to $f\leq t$ as the actual number of Byzantine
processes.  A Byzantine process is called {\it faulty}, a non-faulty
process is {\it correct}.  

\myparagraph{Algorithm semantics} The asynchronous semantics of a
distributed algorithm executed by processes in $\Pi$ assumes discrete
time and at each point in time, exactly one process takes a step. We
assume that two messages cannot be received at the same time by the
same process. The global execution then consists in an interleaving of
the individual steps taken by the processes.  Process $p_i$ sends a
message to $p_j$ by invoking the primitive ``$\lit{send}$ {\sc
  header}$(m)$ to $p_j$'', where {\sc header} indicates the type of
message and $m$ its contents. Process $p_i$ receives a message by
executing the primitive ``$\lit{receive}()$''. The shorthand
$\lit{broadcast}$({\sc
  header}$, m)$ represents ``for each $p_j \in \Pi$ do $\lit{send}$ 
  {\sc header}$(m)$ to $p_j$''.
  And the right arrow in $\lit{broadcast}$({\sc
  header}$, m) \to \ms{messages}$ indicates, when specified, 
  %
that ``upon reception of {\sc header}$(m)$ from process $p_j'$ 
do
$\ms{messages}[p_j'] \gets \ms{messages}[p_j'] \cup \{m\}$''.  The
process id is used as a subscript to denote that a variable is local
to a process---for instance $\ms{var}_i$ is local to process $p_i$---and
is omitted when it is clear from the context.

The verification method considered in this paper exploits
the fact that the algorithms are communication-closed~\cite{EF82}, \emph{i.e.}
only messages from the current loop iteration or \emph{round} of a process may influence its
steps. This can be implemented by tagging every message by its round
number $r$; during round $r$ all received messages with tag $r'< r$
are discarded and all received messages with tag $r'>r$ are stored for
later.  

\myparagraph{The consensus problem} Assuming that each correct process
proposes a binary value, the binary Byzantine consensus problem is for
each of them to decide on a binary value in such a way that the
following properties are satisfied:
\begin{enumerate}
    \item Termination. Every correct process eventually decides on a value.
    \item Agreement.   No two correct processes decide on different values.
    \item Validity.  If all correct processes propose the same value, no
          other value can be decided.
\end{enumerate}



\myparagraph{Threshold automaton (TA)}
A \emph{threshold automaton}~\cite{KVW17} describes the behavior of a
process in a distributed algorithm. Its nodes are \emph{locations}
representing local states, and labeled edges are \emph{guarded rules}.
Formally, it is a tuple
$\tup{{\mathcal L}, {\mathcal I}, \Gamma, {\mathcal P}, {\mathcal R}, \ms{RC}}$
where ${\mathcal L}$ is the set of locations,
${\mathcal I} \subset {\mathcal L}$ is the set of initial locations,
$\Gamma$ is the set of shared variables that all processes can update, ${\mathcal P}$ is the finite set of
parameter variables, $\mathcal{R}$ is the set of rules, and $\ms{RC}$
is the resilience condition over ${\mathbb N}^{|\Pi|}_0$. Rules are
defined as tuples $\tup{\ms{from}, \ms{to}, \phi, \vec{u}}$, where
$\ms{from}$ (resp. $\ms{to}$) describes the source (resp. destination)
locations, and the rule label is $\phi \mapsto \vec{u}$.
Formula~$\phi$ is called a \emph{threshold guard} or simply a
\emph{guard}.

\begin{figure}[htbp]
  \centering
\begin{algorithmic}[1]
                %
                %

                \Part{$\lit{bv-broadcast}(\lit{BV}, \tup{\ms{val}, i})$}{ \Comment{broadcast binary value $\ms{val}$} \label{line:bcast-1}
                  \State $\lit{broadcast}(\lit{BV}, \tup{\ms{val}, i})${
              }\label{line:bvb-first-bcast} \Repeat
                  \Comment{re-broadcast a received value only if it is
                    sufficiently represented}\label{line:repeat-start}
                  \If{$(\lit{BV}, \tup{v, *})$ received from $(t + 1)$
                    distinct processes
                    but \label{line:bcast-2}  not yet
                    re-broadcast} \label{line:no-rebcast}

                  \State $\lit{broadcast}(\lit{BV}, \tup{v,
                    i})$ \label{line:bvbcast-rebcast}

                  \EndIf
                  \If{$(\lit{BV}, \tup{v, *})$ received from
                    $(2t + 1)$  distinct
                    processes} \label{line:guard-2} 

                  \State $\binvalues \leftarrow \binvalues \cup
                  \{v\}$ \label{line:BYZ-safe-14} \Comment{deliver $v$}%
                  
                  \EndIf%

                  \EndRepeat \label{line:bvbcast-end}%
                }\EndPart

            \end{algorithmic}%
  \caption{The pseudocode of the binary value broadcast for process $p_i$.}
  \label{fig:bvb-pseudocode}
\end{figure}

\begin{figure*}[ht]
        \resizebox{0.7\textwidth}{!}{
        \tikzstyle{node}=[circle,draw=black,thick,minimum size=4.3mm,inner sep=0.75mm,font=\large]
        \tikzstyle{init}=[node,fill=yellow!10]
        \tikzstyle{final}=[node,fill=blue!10, accepting]
        \tikzstyle{rule}=[->,thick]
        \tikzstyle{post}=[->,thick,rounded corners,densely dashed,font=\large]
        \tikzset{every loop/.style={min distance=5mm,in=75,out=108,looseness=2}}
        \begin{tikzpicture}[>=latex, thick, scale=0.85, every node/.style={scale=0.8},node distance=1.6cm]
            \node at (0,0) [init,label=below:\textcolor{blue}{$V_0$}]    (V0) {};
            \node[below = of V0] [init,label=below:\textcolor{blue}{$V_1$}]   (V1) {};
            \node[right = 0.8cm of V0] [node,label=below:\textcolor{blue}{$B_0$}]   (B0) {};
            \node[right = 0.8cm of V1] [node,label=below:\textcolor{blue}{$B_1$}]   (B1) {};
            \node[right = 5cm of B0] [final,label=below:\textcolor{blue}{$CB_0$}]   (CB0) {};
            \node[right = 5cm of B1] [final,label=below:\textcolor{blue}{$CB_1$}]   (CB1) {};
            \node at ($(B0)!0.5!(CB1)$) [node,label=below:\textcolor{blue}{$B_{01}$}]   (B01) {};
            \node[above = of B01] [final,label=below:\textcolor{blue}{$C_{0}$}]   (C0) {};
            \node[below = of B01] [final,label={[label distance=-1pt]100:\textcolor{blue}{$C_1$}}]   (C1) {};
            \node[right = 5cm of B01] [final,label=below:\textcolor{blue}{$C_{01}$}]   (C01) {};
            \draw[post] (V0) to[]
            node[sloped, above, align=center] {$r_1\colon \cpp{b_0}$} (B0);
            \draw[post] (V1) to[]
            node[sloped, above, align=center] {$r_2\colon \cpp{b_1}$} (B1);
            \draw[post] (B0) to[]
            node[sloped, above, align=center] {$b_0 \ge 2t{+}1{-}f$}
            node[sloped, below, align=center] {$r_3$} (C0);
            \draw[post] (B0) to[]
            node[sloped, above, align=center,pos=0.55] {$b_1{\ge}t{+}1{-}f \mapsto \cpp{b_1}$}
            node[sloped, below, align=center] {$r_4$} (B01);
            \draw[post] (B1) to[]
            node[sloped, above, align=center,pos=0.45] {$b_0{\ge}t{+}1{-}f \mapsto \cpp{b_0}$}
            node[sloped, below, align=center] {$r_5$} (B01);
            \draw[post] (B1) to[]
            node[sloped, above, align=center, pos=0.45] {$b_1{\ge} 2t{+}1{-}f$}
            node[sloped, below, align=center] {$r_6$}(C1);
            \draw[post] (C0) to[]
            node[sloped, above, align=center,pos=0.45] {$b_1{\ge}t{+}1{-}f \mapsto \cpp{b_1}$}
            node[sloped, below, align=center] {$r_7$} (CB0);
            \draw[post] (B01) to[]
            node[sloped, above, align=center] {$b_0 \ge 2t{+}1{-}f$}
            node[sloped, below, align=center] {$r_8$} (CB0);
            \draw[post] (B01) to[]
            node[sloped, above, align=center] {$b_1{\ge} 2t{+}1{-}f$}
            node[sloped, below, align=center] {$r_{9}$}(CB1);
            \draw[post] (C1) to[]
            node[sloped, above, align=center,pos=0.45] {$b_0{\ge}t{+}1{-}f \mapsto \cpp{b_0}$}
            node[sloped, below, align=center] {$r_{10}$} (CB1);
            \draw[post] (CB0) to[]
            node[sloped, above, align=center] {$b_1{\ge} 2t{+}1{-}f$}
            node[sloped, below, align=center] {$r_{11}$} (C01);
            \draw[post] (CB1) to[]
            node[sloped, above, align=center] {$b_0 \ge 2t{+}1{-}f$}
            node[sloped, below, align=center] {$r_{12}$} (C01);
            \draw[rule,loop] (V0) to[] (V0);
            \draw[rule,loop] (V1) to[] (V1);
            \draw[rule,loop] (C0) to[min distance=5mm,in=147,out=180,looseness=2] (C0);
            \draw[rule,loop] (C1) to[min distance=5mm,in=10,out=-23,looseness=2] (C1);
            \draw[rule,loop] (CB0) to[] (CB0);
            \draw[rule,loop] (CB1) to[min distance=5mm,in=10,out=-23,looseness=2] (CB1);
            \draw[rule,loop] (C01) to[] (C01);
        \end{tikzpicture}
        }
        \caption{The threshold automaton model for the binary value broadcast.}
        \label{fig:bvb-ta}
\end{figure*}
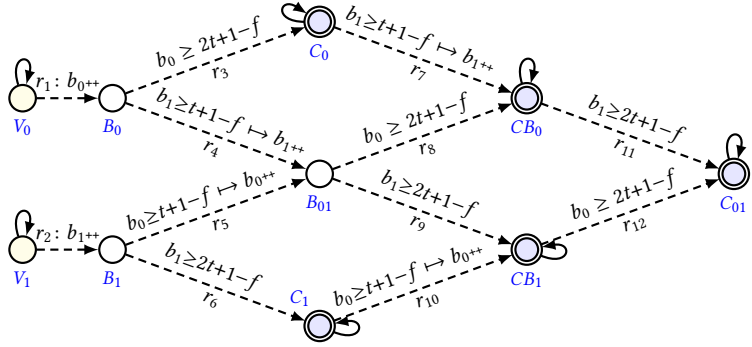

\begin{example}
  As an example, Fig.~\ref{fig:bvb-pseudocode} 
  presents the
  pseudocode of the binary value broadcast and Fig.~\ref{fig:bvb-ta} its TA. 
  (The modeling of pseudocode (Fig.~\ref{fig:bvb-pseudocode}) into TA (Fig.~\ref{fig:bvb-ta}) will be described in detail in Section~\ref{ssec:bvb}.)
  To illustrate the TA notations, note that two of the locations in
  $\mathcal{L} = \{V_0, V_1, B_0, B_1, B_{01}, C_0, C_1, CB_0, CB_1,
  C_{01}\}$ are initial: $\mathcal{I}= \{V_0, V_1\}$.  Shared
  variables are~$b_0$ and~$b_1$ and can be updated by each process traversing the \TA, while parameters are $n$, $t$ and
  $f$ and remain unchanged across the execution.  The set of rules $\mathcal{R}$ consists of
  $\{r_i \mid 1\le i \le 12\}$ together with~7 self-loops. The self-
  loops mimic the asynchrony between processes in the system.
  For
  example, rule $r_3$ is defined as
  $\tup{B_0, C_0, b_0{\ge} 2t{+}1{-}f, \vec{0}}$.  The resilience
  condition is $n>3t \land t\ge f\ge 0$.
\end{example}

A \emph{multi-round threshold automaton} is intuitively defined such
that one round is represented by a threshold automaton, and
additional so-called \emph{round-switch rules} connect
final locations with initial ones, and therefore allow processes to
move from one round to the following one. We typically depict those
round-switch rules as dotted arrows. Examples of such multi-round TA are
depicted later in Figures~\ref{fig:large-ta}
and~\ref{fig:dbft_safe_comp}. When it is clear from the context that
there are multiple rounds, we simply call them threshold automata, and
to stress that a TA does not have multiple rounds, we may call it a
one-round TA.

\myparagraph{Counter systems}
The semantics of a (one-round) threshold automaton $\TA$ are given by a
counter system $\ms{Sys}(\TA)=\tup{\Sigma, I, T}$ where $\Sigma$ is
the set of all configurations among which $I$ are the initial ones,
and $T$ is the transition relation.  A configuration $\sigma\in\Sigma$
of a one-round \TA captures the values of location counters (counting
the number of processes at each location of ${\mathcal L}$, therefore non-negative
integers), values of global variables, and parameter values.  A
transition $t\in T$ is \emph{unlocked in} $\sigma$ if there exists a
rule~$r=\tup{\ms{from}, \ms{to}, \phi, \vec{u}}\in \mathcal{R}$ such
that $\phi$ evaluates to true in $\sigma$, and location counter of
$\ms{from}$ is at least 1, denoted $\kappa[\ms{from}]\ge 1$, showing
that at least one process is currently in $\ms{from}$.  In this case
we can execute transition~$t$ on~$\sigma$ by moving a process along
the rule~$r$ from location $\ms{from}$ to location~$\ms{to}$, which is
modeled by decrementing counter $\kappa[\ms{from}]$, incrementing
$\kappa[\ms{to}]$, and updating global variables according to the
update vector~$\vec{u}$.

A counter system~$\ms{Sys}(\TA)$ of a multi-round $\TA$ is defined analogously.
A configuration captures the values of location counters and global variables \emph{in each round}, and parameter values (that do not change over rounds).
Then we define that a transition is \emph{unlocked in a round}~$R$ by evaluating the guard~$\phi$ and the counter of location~$\ms{from}$ in the round~$R$.
The execution of the transition in~$\sigma$ accordingly updates $\kappa[\ms{from},R]$, $\kappa[\ms{to},R]$ and global variables of that round, while the values of these variables in other rounds stay unchanged.


\myparagraph{Linear temporal logic notations}
Following a standard model checking approach, we use formulas in
linear temporal logic (\LTL)~\cite{Pnueli77} to formalize the desired
properties of distributed algorithms. The basic elements of these
formulas, called atomic propositions, are predicates over
configurations related (i) to the emptiness of each location at each
round and (ii) to the evaluation of threshold guards in each round.
They have the following form: (i) $\kappa[L, R] \neq 0$ expresses that
at least one correct process is in location $L$ in round $R$, while
$\kappa[L, R] = 0$ expresses the opposite (in one-round systems we
just write $\kappa[L]\neq 0$ or $\kappa[L]=0$); (ii) the evaluation of
$[b_0,R] {\ge} 2t{+}1{-}f$ depends on the values of the shared
variable~$b_0$ in round~$R$ and parameters $t$ and $f$ (in one-round
systems we just write $b_0{\ge} 2t{+}1{-}f$).  \LTL\ builds on
propositional logic with $\Rightarrow$ for `implication', $\vee$ for
`or' and $\wedge$ for `and', and has extra temporal operators $\ltlF$
standing for `eventually' and $\ltlG$ for `always'. \LTL\ formulas
are evaluated over infinite runs of $\ms{Sys}(\TA)$.  Examples of
\LTL\ properties in a one-round system are \eqref{eq:bvjust},
\eqref{eq:bvobl} and \eqref{eq:bvunif} (see page~\pageref{eq:bvjust}).
\LTL\ properties in multi-round systems often have quantifiers over
round variables, as for example in \eqref{eq:agree-cc} and
\eqref{eq:valid-cc} (see page~\pageref{eq:agree-cc}).

The tool ByMC is used to automatically verify a specific fragment of
\LTL\ on one-round systems~\cite{KLVW17,KLVW17popl}, which is
sufficient to express safety and liveness properties of
consensus~\cite{BertrandKLW19}. Moreover, thanks to
communication-closure, the verification for this fragment of temporal
logic on multi-round systems reduces to one-round
systems~\cite[Theorem~6]{BertrandKLW19} (see also
Appendix~\ref{app:multi2oneround}).

The assumption of reliable communication is modeled as follows at the
\TA level: if the guard of a rule is true infinitely often, then the
origin location of that rule will eventually be empty. This reflects
that an if branch of the pseudo-code is taken if the condition is
true. This \emph{progress assumption} is in particular crucial to
prove liveness properties: in the sequel, we prepend it to the
liveness properties in the TA specification.

\section{The Binary Value Broadcast}\label{sec:da-models}
To overcome the limited scalability of model checking tools, our
holistic verification approach consists of decomposing a
distributed algorithm into encapsulated components of pseudocode that can be modelled 
in threshold automata and verified in
isolation to obtain a simplified threshold automaton that is amenable
to automated verification.

In this section we focus on a \emph{binary value broadcast}, or
\BVbcast\ for short, that will serve as the main building block of the
Byzantine consensus algorithm of Section~\ref{sec:composition}.  In
Section~\ref{ssec:bvb} we formally model the \BVbcast\ algorithm pseudocode as a threshold
automaton that tolerates a number $f$ of Byzantine failures
upper-bounded by $t$ among $n$ processes.  In
Section~\ref{ssec:bvbcas-proof} we model the specification of
\BVbcast\ in \LTL\ and verify, within 10~seconds, that it holds.  In
Section~\ref{sec:fairness} we introduce the fairness of an infinite
sequence of executions of \BVbcast\ that will play a crucial role in
verifying holistically in Section~\ref{sec:verification} the Byzantine consensus
algorithm.

\subsection{Modeling the binary value broadcast pseudocode into a threshold automaton}\label{ssec:bvb}
The binary value broadcast~\cite{MMR14}, or \emph{bv-broadcast} for
short, is a communication primitive guaranteeing that all binary
values ``bv-delivered'' were ``bv-broadcast'' by a correct process. It
is particularly useful to solve the Byzantine consensus problem with
randomization~\cite{MMR15,CZ20} or partial
synchrony~\cite{CGLR18,CCC20}.  As discussed before,
Figures~\ref{fig:bvb-pseudocode} and~\ref{fig:bvb-ta} in
Section~\ref{sec:notations} give its pseudocode and the corresponding
threshold automaton, respectively.
We now explain how we model our \BVbcast\ pseudocode (Fig.~\ref{fig:bvb-pseudocode}) parameterized by $n$ and $f$ into a threshold automaton (Fig.~\ref{fig:bvb-ta}) using the synthesis methodology~\cite{LWB17}.

\myparagraph{Pseudocode of the binary value broadcast}
The \BVbcast\ algorithm pseudocode (Fig.~\ref{fig:bvb-pseudocode}) aims at
having at least $2t{+}1$ processes broadcasting the same binary value. 
Each process starts this algorithm in one of two states, depending on its input value 0 or 1. 
Once a correct process receives a value from $t{+}1$ distinct processes, it broadcasts it (line~\ref{line:bcast-2}) if it did not do it already (line~\ref{line:no-rebcast}); \bcast\ is not Byzantine fault tolerant and just sends a message to all the other processes. Once a correct process receives a value from $2t{+}1$ distinct processes, it delivers it. Here the delivery at process $p_i$ is modeled by adding the value to the set $\ms{contestants}$, which will simplify the pseudocode of the Byzantine consensus algorithm in Section~\ref{ssec:bbc-safe}.
%

\myparagraph{Threshold automaton of the binary value broadcast}
To match the two initial states from which a process starts the  \BVbcast\ algorithm, we start the corresponding 
TA of Fig.~\ref{fig:bvb-ta} with 
two
initial locations $V_0$ or $V_1$, indicating whether the (correct)
process initially has value 0 or 1, resp.
We can see form the pseudocode (Fig.~\ref{fig:bvb-pseudocode}) that a correct process $p_i$ sends only two types of messages, $(\lit{BV}, \langle 0, i \rangle)$ and $(\lit{BV}, \langle 1, i \rangle)$, these trigger the corresponding receptions at other processes. We thus define in the TA (Fig.~\ref{fig:bvb-ta}) two global variables $b_0$ and $b_1$, resp., to capture the number of the two types of messages sent by correct processes. Thus, for example, $\cpp{b_0}$ models a process broadcasting message $(\lit{BV}, \langle 0, i \rangle)$.
Because the algorithm only counts messages regardless of sender identities, we replace the messages from the pseudocode into $b_0$ and $b_1$ shared variables that are increased whenever a message is sent.

\myparagraph{From local to global variables for model checking}
While producing a formal model, extra care is needed to avoid introducing redundancies.
For example, line~\ref{line:bcast-2} indicates that the process broadcasts value $v$ if it received $v$ from $t{+}1$ distinct processes.
Instead of maintaining local receive variables, it is sufficient to enable a guard based on global send variables.
Indeed, to remove redundant local receive variables, one can use the quantifier elimination for Presburger arithmetic~\cite{Pres29} and obtain quantifier-free guard expressions over the shared variables that are valid inputs to ByMC~\cite{KW18,marilinajogor}.
For more details, note that Stoilkovska et al.~\cite{SK20} eliminated the quantifier over the similar receive variables in Ben-Or's consensus algorithm~\cite{Ben83} with the  SMT solver Z3~\cite{MB-tacas08}.
Finally, the point-to-point reliable channels ensure that
$p_j$ sends message $m$ to $p_i$ implies that eventually $p_i$ receives message $m$ from $p_j$.
Hence shared variables $b_0$ and $b_1$ of the TA denote, respectively, the number of messages $(\lit{BV}, \langle 0, i \rangle)$ and $(\lit{BV}, \langle 1, i \rangle)$ \emph{sent} by correct processes in the pseudocode.

\myparagraph{Modeling arbitrary (Byzantine) behaviors in the \TA}
In order to model that, among the received messages, $f$ messages could have been sent
by Byzantine processes,
%
we need to map the `if' statement of the pseudocode, comparing the number of receptions from distinct processes to $t{+}1$,
to the \TA\ guards, comparing the number $b_1{+}f$ of messages sent to $t+1$.
As $b_1$ counts the messages \textit{sent} by correct processes and $f$ is the number of faulty processes
that can send arbitrary values, a correct process can move from ${B_0}$ to ${B_{01}}$ as soon as $t{+}1{-}f$ correct
processes have sent $1$, provided that $f$ faulty processes have also sent $1$.
%
As a result, the guard of rule $r_4$ only evaluates over global send variables as: if more than $t{+}1$ messages of type $b_1$ have been \emph{sent} by correct processes (hence the guard $b_1 \geq t{+}1{-}f$), then the shared variable $b_1$ is incremented, mimicking the \emph{broadcast} of a new message of type $b_1$.
Rule $r_3$ corresponds to lines~\ref{line:guard-2}--\ref{line:BYZ-safe-14} and delivers value $v=0$ by storing it into variable $\ms{contestants}$ upon reception of this value from $2t+1$ distinct processes.
%
Hence, reaching location ${C_0}$ in the \TA\ indicates that the value $0$ has been delivered.
As a process might stay in this location forever, we add a self-loop with guard condition set
to $\lit{true}$.

%
        \begin{table}[H]
        \setlength{\tabcolsep}{3pt}
            \begin{center}
                \begin{tabular}{ l | c c c c c c c c c c }
                    \toprule
                    locations      & $V_0$ & $V_1 $ & $B_0$ & $B_1$ & $B_{01}$ & $C_0$ & $CB_0$ & $C_1$ & $CB_1$ & $C_{01}$ \\
                    \midrule
                    val. broadcast & /     & /      & 0     & 1     & 0,1      & 0     & 0,1    & 1     & 0,1    & 0,1      \\
                    val. delivered & /     & /      & /     & /     & /        & 0     & 0      & 1     & 1      & 0,1      \\
                    \bottomrule
                \end{tabular}
                \vspace{1em}
                \caption{The locations of correct processes\label{table:loc}}
            \end{center}
        \end{table}

\myparagraph{Other locations and rules}
The locations of the automaton correspond to the
exclusive situations for a correct process depicted in Table~\ref{table:loc}.
After  location ${C_0}$, a process is still able to broadcast $1$ and
eventually deliver $1$ after that. After  location ${B_{01}}$, a
process is able to deliver $0$ and then deliver $1$, or deliver $1$
first and then deliver $0$, depending on the order in which the guards
are satisfied.
Apart from the self-loops, note that the automaton is a directed acyclic graph. Also, on every path in the graph,  a shared variable is incremented only once. This reflects that in the pseudocode, a value may only be broadcast if it has not been broadcast before.
%
%

\subsection{Properties of the binary value broadcast}\label{ssec:bvbcas-proof}

As was previously proved by hand~\cite{MMR14,MMR15}, the \BVbcast\
primitive satisfies four properties: BV-Justification, BV-Obligation,
BV-Uniformity and BV-Termination. Here, we formalize these properties
in linear temporal logic (\LTL) to formally and automatically prove
they hold. As we will discuss in Section~\ref{sec:expe}, we verify them
for any parameters $n$ and $t<n/3$ in less then 10
seconds. 


The BV-Justification property states:
``If $p_i$ is correct and $v\in \binvalues_i$, then $v$ has been
bv-broadcast by some correct process'' where $v \in \{0,1\}$. Alternatively, ``if $v$ is not
bv-broadcast by some correct process and $p_i$ is correct, then
$v \notin \binvalues_i$''. In the TA from
Fig.~\ref{fig:bvb-ta}, $v\in\binvalues_i$ corresponds to
process $i$ being in one of the locations $C_v$, $CB_v$ or $C_{01}$.
Thus, justification can be expressed in \LTL\ as the conjunction
$\ms{BV-Just_0} \land \ms{BV-Just_1}$ where,
$\ms{BV-Just_v}$ is the following formula:
\begin{equation*}\label{eq:bvjust}
    \tag{$\ms{BV-Just_v}$}
    \kappa[V_v] =0 \;\Rightarrow\; \ltlG \bigl(\kappa[C_v] = 0 \wedge \kappa[CB_v] = 0 \wedge \kappa[C_{01}] = 0\bigr) \enspace.
\end{equation*}


BV-Obligation requires that if at least $(t{+}1)$ correct  processes bv-broadcast the same value $v$,
then~$v$ is eventually added to the set $\ms{contestants}_i$ of each correct process $p_i$.
This can again be formalized as $\ms{BV-Obl_0} \land \ms{BV-Obl_1}$
where $\ms{BV-Obl_v}$ is the following formula:
\begin{equation*}\label{eq:bvobl}
    \tag{$\ms{BV-Obl_v}$}
    \ltlG\Bigl( b_v\ge t{+}1 \;\Rightarrow\; \ltlF\bigl(\bigwedge_{L\in \text{Locs}_v}\kappa[L] = 0\bigr)\Bigr) \enspace,
\end{equation*}
where $\text{Locs}_v=\{V_0, V_1, B_0, B_1, B_{01}, C_{1-v}, CB_{1-v}\}$ are all the possible locations of
a process~$i$ if $v\not\in\binvalues_i$.

BV-Uniformity requires that if a value~$v$ is added to the set $\ms{contestants}_i$ of a correct process $p_i$,
then eventually  $v\in \ms{contestants}_j$ at every correct process $p_j$.
We formalize this as $\ms{BV-Unif_0} \land \ms{BV-Unif_1}$
where $\ms{BV-Unif_v}$ is the following:
\begin{equation*}\label{eq:bvunif}
    \tag{$\ms{BV-Unif_v}$}
    \ltlF(\kappa[C_v]{\ne}0 ~\vee~ \kappa[CB_v]{\ne}0 ~\vee~ \kappa[C_{01}]{\ne}0)
    \;\Rightarrow\;
    \ltlF \bigwedge_{L\in\text{Locs}_v} \kappa[L]{=}0 \enspace,
\end{equation*}
where $\text{Locs}_v$ is defined as in $\ms{\eqref{eq:bvobl}}$.

Finally, the BV-Termination property claims that eventually the set $\ms{contestants}_i$ of
each correct process $p_i$  is non empty. This can be phrased as the following \LTL\ formula $\ms{BV-Term}$:
\begin{equation*}
    \tag{$\ms{BV-Term}$}
    \ltlF \bigl( \kappa[V_0]{=}0 ~\wedge~ \kappa[V_1]{=}0 ~\wedge~ \kappa[B_0]{=}0 ~\wedge~ \kappa[B_1]{=}0 ~\wedge~ \kappa[B_{01}]{=}0 \bigr) \enspace,
\end{equation*}
forcing each correct process to be in one of the ``final'' locations $C_0, C_1, C_{01}, CB_0, CB_1$.

\medskip

\subsection{A fairness assumption to solve consensus}\label{sec:fairness}

The traditional approach to establishing guarantee properties in verification 
is to require that all fair computations, instead of all computations, satisfy the property~\cite{AH94}. 
We thus introduce the fairness assumption 
that will be crucial in the rest of this paper.
In order to define it, we first define a \emph{good execution} of the
\BVbcast\ with respect to binary value $v$ as an execution:


\begin{definition}[$v$-good \BVbcast]\label{def:good-exec}
    A \BVbcast\ execution is $v$-\emph{good} if all its correct processes bv-deliver $v$ first.
\end{definition}
We express this property in \LTL.
A \BVbcast\ execution is $v$-\emph{good} if	no process ever visits locations~$C_{1-v}$ and~$CB_{1-v}$:
\begin{equation*}
    \ltlG \;\Big( \kappa[C_{1-v}]=0 ~\land~ \kappa[CB_{1-v}]=0  \Big)\enspace.
\end{equation*}

Second, we consider an infinite sequence of \BVbcast\ executions,
tagged with $r \in \nats$.
It is important to stress that the setting is asynchronous, that is,
processes invoke \BVbcast\ infinitely many times, but at their own relative speed.
Thus, they do not all invoke the \BVbcast\ tagged with the same number~$r$ at the same time.
Nonetheless, every process invokes \BVbcast\ infinitely many times and
in the $r^{\ms{th}}$ invocation its behavior depends
on the messages sent in the $r^{\ms{th}}$ invocation of other processes.
Therefore, we refer to the $r^{\ms{th}}$ execution of \BVbcast\ even though the
processes invoke it at different times.

\begin{definition}[fairness]\label{def:fair-bvb}
    An infinite sequence of \BVbcast\ executions is \emph{fair} if
    there exists an~$r$ such that the $r^{\ms{th}}$ execution is
    $(r \modulo 2)$-good.
\end{definition}
For simplicity, we use the terminology \emph{fair \BVbcast} when the
infinite sequence of \BVbcast\ executions is fair.
We illustrate in Appendix~\ref{app:no-termination} a possible
execution of \BVbcast\ whose existence implies fairness.

\section{Simplified Automaton for Byzantine Consensus}\label{sec:composition}

In this section we exploit the results of the first verification phase
of Section~\ref{sec:da-models} to simplify the threshold automaton of
the Byzantine consensus algorithm.  In Section~\ref{ssec:bbc-safe} we
introduce the pseudocode of the Byzantine consensus algorithm and its
threshold automaton obtained with the naive 
modeling described in
Section~\ref{ssec:bvb}. 
In Section~\ref{sec:states-ta} we replace, in this threshold
automaton, the inner \BVbcast\ automaton by a smaller one obtained
thanks to the \BVbcast\ properties that are now verified.  The
verification of the resulting simplified automaton is deferred to
Section~\ref{sec:verification}.

\subsection{The Byzantine consensus algorithm}
\label{ssec:bbc-safe}

Algorithm~\ref{alg:bbc-safe-cc} is the DBFT Byzantine consensus algorithm~\cite{CGLR18}
that relies on the fair binary value broadcast of Section~\ref{sec:da-models}.
It is currently used in the Red Belly Blockchain, a recent blockchain that achieves unprecedented scalability~\cite{CNG21}.
More precisely, the DBFT binary consensus comes in two different variants: (i)~a first variant that is safe but not live in the asynchronous setting, (ii)~a second variant that is safe and live under the partial synchrony assumption.
We use the first variant of it (without coordinator or timeout) here and show that it is live under our new fairness assumption.
The DBFT binary consensus invokes $\lit{bv-broadcast}(\cdot)$ at line~\ref{BYZ-safe-04} and uses a set  $\ms{contestants}$ of binary values, whose scope is global, updated by the $\lit{bv-broadcast}$ (Fig.~\ref{fig:bvb-pseudocode}, line~\ref{line:BYZ-safe-14})
and accessed by the procedure $\lit{propose}(\cdot)$ (Alg.~\ref{alg:bbc-safe-cc}, line~\ref{line:BYZ-safe-05}). 

        \begin{small}

\begin{algorithm}[H]
	\caption{The Byzantine consensus algorithm at process $p_i$\label{alg:bbc-safe-cc}}
	\begin{algorithmic}[1]
	
			\Part{Global scope variable}{
				\State $\ms{contestants} \subseteq \{0,1\}$, set of binary values, init. $\emptyset$.
			}\EndPart
			
			\Statex
			
%
%
%
	
%
	
		\Part{$\lit{propose}(\ms{est})$} { \label{line:safe-binpropose}
			\State $r \gets 0$ \label{line:set-rnumber}
			\Repeat \label{line:round-start} 
			\State $\BVbcast(\EST, \tup{\ms{est}, i})$ \label{line:bvbcast-in-consensus} \label{BYZ-safe-04}
			\WUntil$(\binvalues \neq \emptyset$) \label{line:BYZ-safe-05}
				\Comment{wait enough time to receive some $\binvalues$}\EndWUntil
			\State $\lit{broadcast}(\AUX, \tup{\binvalues, i}) \to \favorites$ \label{line:BYZ-safe-06}  \Comment{broadcast second phase message}\label{line:bincons:avcauxbcast}
			
			\WUntil $\exists c_1, \dots, c_{n-t} : \forall 1\leq j\leq n-t\ \ms{favorites}[c_j] \neq \emptyset \,\wedge$\label{BYZ-safe-07-a} $(\ms{qualifiers} \gets \cup_{\forall 1\leq j\leq n-t}\,\ms{favorites}[c_j]) \subseteq \binvalues$
				\EndWUntil \label{BYZ-safe-07} 
			
			\If{$\ms{qualifiers} = \{v\}$} 
\label{BYZ-safe-09} \Comment{if there is only one value, then adopt it}  \label{BYZ-safe-08}
				\State $\ms{est} \leftarrow v$  \Comment{estimate becomes this singleton value} \label{BYZ-safe-11}
				\If{$v = (r \modulo 2)$}  $\lit{decide}(v)$ \label{BYZ-safe-10} \Comment{decide once if value matches parity} 
				\EndIf
				\Else{}
					 $\ms{est} \leftarrow (r \modulo 2)$   \Comment{otherwise, adopt the current parity bit} \label{BYZ-safe-12} 
			\EndIf
			\State $r \leftarrow r + 1$ \label{line:inc-rnumber} \Comment{increment the round number}
			\EndRepeat \label{line:round-end}
		}\EndPart
		
		\end{algorithmic}
\end{algorithm}
	\end{small}
        
As mentioned in Section~\ref{sec:notations}, recall that the algorithm is communication-closed, so that for simplicity in the presentation we omit the current round number $r$ as the subscript of the variables and the parameter of the function calls. Variable $\ms{favorites}$ is an array of $n$ indices whose $j^{\ms{th}}$ slot records, upon delivery, the message broadcast by process $j$ in the current round.
Each process $p_i$ manages the following local variables: the current estimate $\ms{est}$, initially the input value of $p_i$;
and a set of binary values $\ms{qualifiers}$.
This algorithm
maintains a round number $r$, initially 0 (line~\ref{line:set-rnumber}), and incremented
at the end of each iteration of the loop at line~\ref{line:inc-rnumber}.
%
Process
$p_i$ exchanges \EST and \AUX messages (lines~\ref{BYZ-safe-04}--\ref{line:BYZ-safe-06}), until it receives \AUX messages from $n-t$ distinct processes whose values were bv-delivered by $p_i$ (line~\ref{BYZ-safe-07-a}).
Process $p_i$ then tries at line~\ref{BYZ-safe-10} to decide a value $v$ that depends on the content of
$\ms{qualifiers}$ and the parity of the round.
If
$\ms{qualifiers}$
is a singleton there are two possible cases: if the value is the parity of the round then $p_i$ decides this value (line~\ref{BYZ-safe-10}), otherwise it sets its estimate to this value (line~\ref{BYZ-safe-11}). If $\favorites$ contains both binary values, then $p_i$ sets its estimate to the parity of the round (line~\ref{BYZ-safe-12}).
Although $p_i$ does not exit the infinite loop to help other processes decide, it can safely exit the loop after two rounds at the end of the second round that follows the first decision because all processes will be guaranteed to have decided. Note that even though a process may invoke $\lit{decide}(\cdot)$ multiple times at line~\ref{BYZ-safe-10}, only the first decision matters as the decided value does not change (see Section~\ref{sec:verification}).
\begin{figure*}[t]
    \begin{center}
        \makebox[\textwidth][c]{\tikzstyle{node}=[circle,draw=black,thick,minimum size=4.3mm,inner sep=0.75mm,font=\normalsize]
\tikzstyle{init}=[node,fill=yellow!10]
\tikzstyle{final}=[node,fill=blue!10]
\tikzstyle{rule}=[->,thick]
\tikzstyle{post}=[->,thick,rounded corners,font=\normalsize]
\tikzstyle{postbvb}=[post,densely dashed]  
\tikzstyle{postdash}=[->,thick,loosely dotted,rounded corners,font=\normalsize]
\tikzset{every loop/.style={min distance=5mm,in=-25,out=18,looseness=2}}

\begin{tikzpicture}[>=latex, thick,node distance=0.9cm, scale=0.65, every node/.style={scale=0.7},xscale=0.7]



    \node at (0,0) [init,label=below:\textcolor{blue}{$V_0$}]    (V0) {};
    \node[below = of V0] [init,label=below:\textcolor{blue}{$V_1$}]   (V1) {};

    \node[right = 0.8cm of V0] [node,label=below:\textcolor{blue}{$B_0$}]   (B0) {};
    \node[right = 0.8cm of V1] [node,label=below:\textcolor{blue}{$B_1$}]   (B1) {};

    \node[node,label=below:\textcolor{blue}{$B_{01}$}] at ($(B0)!0.5!(B1) + (2,0)$) (B01) {};
    \node[above = of B01] [node,label=below:\textcolor{blue}{$C_0$}]   (C0) {};
    \node[below = of B01] [node,label=above:\textcolor{blue}{$C_1$}]   (C1) {};

    \node[right =2.1cm of B0] [node,label={above:\textcolor{blue}{$CB_0$}}]   (CB0) {}; 
    \node[right =2.1cm of B1] [node,label={above:\textcolor{blue}{$CB_1$}}]   (CB1) {}; 

    \node[right = 2.5cm of B01] [node,label=above:\textcolor{blue}{$C_{01}$}]   (C01) {};

    \node[right = of C01] [node,label=below:\textcolor{blue}{$E_1$}]   (E1) {};
    \node[above = of E1] [node,label=below:\textcolor{blue}{$E_0$}]   (E0) {};
    \node[below = of E1] [node,label=above:\textcolor{blue}{$D_1$}]   (D1) {};


    \node[right = 0.7cm of E0] [node,label=below:\textcolor{blue}{$V_0'$}]   (V0') {};
    \node[right = 0.7cm of E1] [node,label=below:\textcolor{blue}{$V_1'$}]   (V1') {};

	\node[right = 0.8cm of V0'] [node,label=below:\textcolor{blue}{$B_0'$}]   (B0') {};
	\node[right = 0.8cm of V1'] [node,label=below:\textcolor{blue}{$B_1'$}]   (B1') {};
	
    \node[node,label=below:\textcolor{blue}{$B_{01}'$}] at ($(B0')!0.5!(B1') + (2,0)$) (B01') {};
    \node[above = of B01'] [node,label=below:\textcolor{blue}{$C_0'$}]   (C0') {};
    \node[below = of B01'] [node,label=below:\textcolor{blue}{$C_1'$}]   (C1') {};

    \node[right =2.1cm of B0'] [node,label={above:\textcolor{blue}{$CB_0'$}}]   (CB0') {};
    \node[right =2.1cm of B1'] [node,label={above:\textcolor{blue}{$CB_1'$}}]   (CB1') {};

    \node[right = 2.5cm of B01'] [node,label=above:\textcolor{blue}{$C_{01}'$}]   (C01') {};

    \node[right = of C01'] [node,label=below:\textcolor{blue}{$E_0'$}]   (E0') {};
    \node[above = of E0'] [node,label=below:\textcolor{blue}{$D_0$}]   (D0) {};
    \node[below = of E0'] [node,label=below:\textcolor{blue}{$E_1'$}]   (E1') {};



    \draw[postbvb] (V0) to[]
    node[sloped, above, align=center] {$r_1\colon \cpp{b_0}$} (B0);
    \draw[postbvb] (V1) to[]
    node[sloped, above, align=center] {$r_2\colon \cpp{b_1}$} (B1);

    \draw[postbvb] (B0) to[]
    node[sloped, anchor=south] {$r_3$} (C0);
    \draw[postbvb] (B0) to[]
    node[sloped, above, align=center] {$r_4$} (B01);
    \draw[postbvb] (B1) to[]
    node[sloped, below, align=center] {$r_5$} (B01);
    \draw[postbvb] (B1) to[]
    node[sloped,anchor=north] {$r_6$} (C1);

    \draw[post] (C0) to[]
    node[sloped, above, align=center] {$r_{13}$} (E0);
    \draw[postbvb] (C0) to[]
    node[sloped, below, align=center] {$r_{7}$} (CB0);

    \draw[postbvb] (B01) to[]
    node[sloped, below, align=center] {$r_{8}$} (CB0);
    \draw[postbvb] (B01) to[]
    node[sloped, below, align=center] {$r_{9}$} (CB1);

    \draw[postbvb] (C1) to[]
    node[sloped, above, align=center] {$r_{10}$} (CB1);

    \draw[postbvb] (CB0) to[]
    node[sloped, above, align=center] {$r_{11}$} (C01);
    \draw[postbvb] (CB1) to[]
    node[sloped, above, align=center] {$r_{12}$} (C01);

    \draw[post] (C1) to[]
    node[sloped, below, align=center] {$r_{19}$} (D1);

    \draw[post] (CB0) to[]
    node[sloped, above, align=center] {$r_{14}$} (E0);
    \draw[post] (CB1) to[]
    node[sloped, above, align=center] {$r_{18}$} (D1);

    \draw[post] (C01) to[] node[sloped, above, pos=0.45] {$r_{15}$} (E0);
    \draw[post] (C01) to[] node[sloped, above, pos=0.5] {$r_{16}$} (E1);
    \draw[post] (C01) to[] node[sloped, above, pos=0.45] {$r_{17}$} (D1);



    \draw[postbvb] (V0') to[]
	node[sloped, above, align=center] {$r_1'\colon \cpp{b_0'}$} (B0');
	\draw[postbvb] (V1') to[]
	node[sloped, above, align=center] {$r_2'\colon \cpp{b_1'}$} (B1');
	
    \draw[postbvb] (B0') to[]
    node[sloped, anchor=south] {$r_3'$} (C0');
    \draw[postbvb] (B0') to[]
    node[sloped, above, align=center] {$r_4'$} (B01');
    \draw[postbvb] (B1') to[]
    node[sloped, below, align=center] {$r_5'$} (B01');
    \draw[postbvb] (B1') to[]
    node[sloped,anchor=north] {$r_6'$} (C1');

    \draw[post] (C0') to[]
    node[sloped, above, align=center,yshift=-0.1cm] {$r_{13}'$} (D0);

    \draw[postbvb] (C0') to[]
    node[sloped, below, align=center] {$r_{7}'$} (CB0');

    \draw[postbvb] (B01') to[]
    node[sloped, below, align=center] {$r_{8}'$} (CB0');
    \draw[postbvb] (B01') to[]
    node[sloped, below, pos=0.4] {$r_{9}'$} (CB1');

    \draw[postbvb] (C1') to[]
    node[sloped, above, pos=0.4] {$r_{10}'$} (CB1');




    \draw[post] (C1') to[]
    node[sloped, above, align=center] {$r_{19}'$} (E1');

    \draw[postbvb] (CB0') to[]
    node[sloped, above, align=center] {$r_{11}'$} (C01');
    \draw[postbvb] (CB1') to[]
    node[sloped, above, align=center] {$r_{12}'$} (C01');

    \draw[post] (CB0') to[]
    node[sloped, above, align=center] {$r_{14}'$} (D0);
    \draw[post] (CB1') to[]
    node[sloped, above, align=center] {$r_{18}'$} (E1');

    \draw[post] (C01') to[] node[sloped, above, pos=0.45] {$r_{15}'$} (D0);
    \draw[post] (C01') to[] node[sloped, above, pos=0.5] {$r_{16}'$} (E0');
    \draw[post] (C01') to[] node[sloped, above, pos=0.45] {$r_{17}'$} (E1');




    %

    
    \draw[post] (E0) to[] node[above]{$r_{20}$} (V0');
    \draw[post] (E1) to[] node[above]{$r_{21}$} (V1');
    \draw[post] (D1) to[] node[sloped,above]{$r_{22}$} (V1');

    \draw[postdash] (D0.north) -- ($(D0)+(0,.5)$) -| ($(V0.north west)+(-0.6,0)$) -- (V0.north west);
    \draw[postdash] (E0'.east) -|($(D0)+(0.75,0.75)$) -| ($(V0.west)+(-0.8,0)$) -- (V0.west);

    \draw[postdash] (E1'.east) -|($(E1')+(0.75,-1.5)$) -| ($(V1.west)+(-0.8,0)$) -- (V1.west);

\end{tikzpicture}}
        \caption{The naive threshold automaton of the Byzantine consensus of Algorithm~\ref{alg:bbc-safe-cc} where the embedded \BVbcast\ automaton is depicted with dashed arrows. Precise formulations of all rules are in Appendix~\ref{app:largeta}.
            Note that the rules $r_{20}, r_{21}$ and $r_{22}$ represent transitions from the end of an odd round to the beginning of the following (even) round of Algorithm~\ref{alg:bbc-safe-cc}, while the dotted edges represent transitions from the end of an even round to the beginning of the following (odd) one.}
        \label{fig:large-ta}
    \end{center}
\end{figure*}
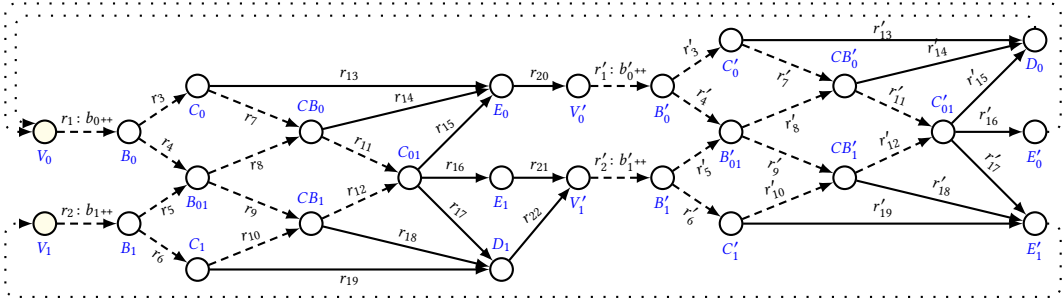

\myparagraph{The effect of fairness}
Note that the fairness notion from Section~\ref{sec:fairness} ensures there is a round~$r$
in which all correct processes bv-deliver $(r \modulo 2)$ first.
The following lemma states that under the fairness assumption there is a round of Algorithm~\ref{alg:bbc-safe-cc} in which all correct processes start with the same estimate.
The proof is deferred to Appendix~\ref{app:same-est}.

\begin{lemma}\label{lem:sameest}
    If the infinite sequence of bv-broadcast executions of Algorithm~\ref{alg:bbc-safe-cc} is fair,
    with the $r^{\ms{th}}$ execution being $(r \modulo 2)$-good,
    then all correct processes start round~$r{+}1$ of
    Algorithm~\ref{alg:bbc-safe-cc} with estimate $r \modulo 2$.
\end{lemma}

\myparagraph{Modeling deterministic consensus}
Figure~\ref{fig:large-ta} depicts the threshold automaton (TA)
obtained by modeling Algorithm~\ref{alg:bbc-safe-cc} with the
method we detailed in Section~\ref{ssec:bvb}.
The TA depicts two iterations
of the repeat loop (line~\ref{line:round-start}), since
Algorithm~\ref{alg:bbc-safe-cc} favors different values depending on the parity of the round number.
For simplicity, we refer to the concatenation of
two consecutive rounds of the algorithm as a \emph{superround} of the
TA.  As one can expect, this TA embeds the TA of the \BVbcast\, which
is depicted by the dashed arrows, just as
Algorithm~\ref{alg:bbc-safe-cc} invokes the $\BVbcast$ algorithm of
Fig.~\ref{fig:bvb-pseudocode}.  We thus distinguish the outer TA
modeling the consensus algorithm from the inner TA modeling the
$\BVbcast$ algorithm.
Although Algorithm~\ref{alg:bbc-safe-cc} is relatively simple, the
global TA happens to be too large to be verified through model
checking, as we explain in Section~\ref{sec:expe};
the main limiting factor is its 14~unique guards that constrain the
variables to enable rules in the TA. The detail of each rule of the TA is deferred to Appendix~\ref{app:largeta}.

\subsection{Simplified threshold automaton}
\label{sec:states-ta}

Our objective is to formally prove that Algorithm~\ref{alg:bbc-safe-cc}
is unconditionally safe, and that it is live under the assumption of
\fairness at the \BVbcast\ level.
Since the threshold automaton of Figure~\ref{fig:large-ta} is too
large to be handled automatically, 
we build on the properties proved for the
$\BVbcast$ 
to simplify in the threshold automaton from Figure~\ref{fig:large-ta} the
part representing the $\BVbcast$. On the resulting simpler threshold
automaton, assuming fairness of the $\BVbcast$, we prove the
termination of Algorithm~\ref{alg:bbc-safe-cc} with the Byzantine
model checker ByMC in Section~\ref{sec:expe}.

%
%


\myparagraph{High-level idea}
%
%
%
%
%
Ideally, the simplified threshold automaton could be obtained from the
one of Fig.~\ref{fig:large-ta} by merging all internal states of the
$\BVbcast$ into a single state with two possible outcomes.  However,
such a merge is not trivial because the $\BVbcast$ procedure
``leaks'' into the consensus algorithm. First of all, line~\ref{line:BYZ-safe-05} of
Algorithm~\ref{alg:bbc-safe-cc} refers to contestants, a global variable that is modified by  the
$\BVbcast$ algorithm (Fig.~\ref{fig:bvb-pseudocode}). Second, a
process can execute line~\ref{line:BYZ-safe-06} of Algorithm~\ref{alg:bbc-safe-cc} even if
the $\BVbcast$ has not terminated. To capture this porosity, we
introduce a new shared variable, some additional states and a
transition rule that exploits a correctness property of the
$\BVbcast$.

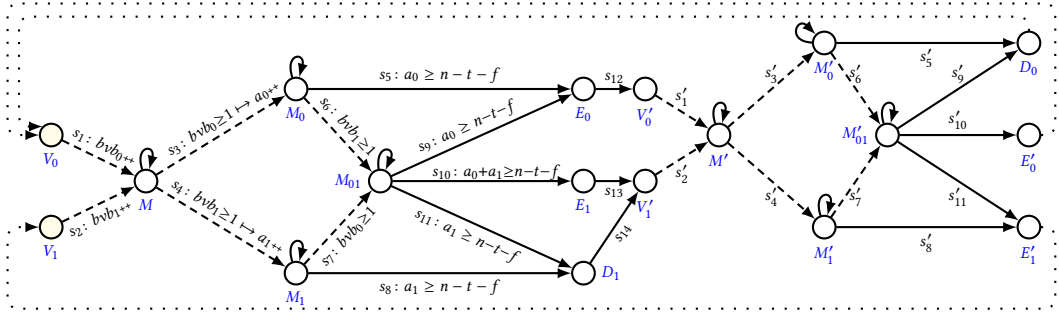
\begin{figure*}[t]
    \begin{center}
        \makebox[\textwidth][c]{\tikzstyle{node}=[circle,draw=black,thick,minimum size=4.3mm,inner sep=0.75mm,font=\normalsize]
\tikzstyle{init}=[node,fill=yellow!10]
\tikzstyle{final}=[node,fill=blue!10]
\tikzstyle{rule}=[->,thick]
\tikzstyle{post}=[->,thick,rounded corners,font=\normalsize]
\tikzstyle{postbvb} = [post, densely dashed]
\tikzstyle{postdash}=[->,thick,loosely dotted,rounded corners,font=\normalsize]

\tikzset{every loop/.style={min distance=5mm,in=75,out=108,looseness=2}}

\begin{tikzpicture}[>=latex, thick,node distance=3.5cm, scale=0.7, every node/.style={scale=0.7}]


    \node at (0,0) [init,label=below:\textcolor{blue}{$V_0$}]    (V0) {};
    \node [init,label=below:\textcolor{blue}{$V_1$}] at ($(V0)+(0,-1.75)$)  (V1) {};

    \node[node,label=below:\textcolor{blue}{$M$}] at ($(V0)!0.5!(V1) + (1.8,0)$) (B) {};

    \node[node,label=below:\textcolor{blue}{$M_0$}] at ($(B) + (2.85,1.75)$)  (B0) {};
    \node[node,label=below:\textcolor{blue}{$M_1$}] at ($(B) + (2.85,-1.75)$)  (B1) {};

    \node[node,label=left:\textcolor{blue}{$M_{01}$}] at ($(B0)!0.5!(B1) + (1.6,0)$) (B01) {};

    \node[right = of B0] [node,label=below:\textcolor{blue}{$E_0$}]   (E0) {};
    \node[right = of B1] [node,label=right:\textcolor{blue}{$D_1$}]   (D1) {};
    \node[node,label=below:\textcolor{blue}{$E_1$}]  at ($(E0)!0.5!(D1)$) (E1) {};

    \node[right = 0.5cm of E0] [node,label=below:\textcolor{blue}{$V_0'$}]   (V0') {};
    \node[right = 0.5cm of E1] [node,label=below:\textcolor{blue}{$V_1'$}]   (V1') {};

    \node[node,label=below:\textcolor{blue}{$M'$}] at ($(V0')!0.5!(V1')+ (1.4,0)$) (Bx) {};

    \node[node,label=below:\textcolor{blue}{$M_0'$}] at ($(Bx) + (2,1.75)$)  (B0x) {};
    \node[node,label=below:\textcolor{blue}{$M_1'$}] at ($(Bx) + (2,-1.75)$)  (B1x) {};

    \node[node,label=left:\textcolor{blue}{$M_{01}'$}] at ($(B0x)!0.5!(B1x) + (1.2,0)$) (B01x) {};

    \node[right = 2.4cm of B0x] [node,label=below:\textcolor{blue}{$D_0$}]   (D0x) {};
    \node[right = 2.4cm of B1x] [node,label=below:\textcolor{blue}{$E_1'$}]   (E1x) {};
    \node[node,label=below:\textcolor{blue}{$E_0'$}] at ($(E1x)!0.5!(D0x)$) (E0x) {};


    \draw[postbvb] (V0) to[]
    node[sloped, above, align=center] {$s_1\colon \cpp{\textit{bvb}_0}$} (B);
    \draw[postbvb] (V1) to[]
    node[sloped, below, align=center,xshift=-0.1cm] {$s_2\colon \cpp{\textit{bvb}_1}$} (B);

    \draw[postbvb] (B) to[]
    node[sloped, above, align=center,xshift=2mm] {$s_3\colon \textit{bvb}_0{\ge}1 \mapsto \cpp{a_0}$}
    (B0);
    \draw[postbvb] (B) to[]
    node[sloped, above, align=center] {$s_4\colon \textit{bvb}_1{\ge}1 \mapsto \cpp{a_1}$}
    (B1);

    \draw[post] (B0) to[]
    node[anchor=south] {$s_5\colon a_0\ge n-t-f$} (E0);
    \draw[postbvb] (B0) to[]
    node[sloped, above, align=center] {$s_6\colon \textit{bvb}_1{\ge}1$} (B01);
    \draw[postbvb] (B1) to[]
    node[sloped, below, align=center] {$s_7\colon \textit{bvb}_0{\ge}1$} (B01);
    \draw[post] (B1) to[]
    node[sloped,anchor=north] {$s_8\colon a_1\ge n-t-f$} (D1);

    \draw[post] (B01) to[]
    node[sloped,above,pos=0.45] {$s_9\colon a_0 \ge n{-}t{-}f$} (E0);
    \draw[post] (B01) to[]
    node[sloped,above,pos=0.58,yshift=-0.1cm] {$s_{10}\colon a_0{+}a_1{\ge}n{-}t{-}f$} (E1);
    \draw[post] (B01) to[]
    node[sloped,below,pos=0.45] {$s_{11}\colon a_1 \ge n{-}t{-}f$} (D1);

	\draw[post] (E0) to[] node[anchor=south] {$s_{12}$} (V0');
	\draw[post] (E1) to[] node[anchor=north] {$s_{13}$} (V1');
	\draw[post] (D1) to[] node[sloped,anchor=north] {$s_{14}$} (V1');

    \draw[postbvb] (V0') to[] node[anchor=south] {$s_1'$} (Bx);
    \draw[postbvb] (V1') to[] node[anchor=north] {$s_2'$} (Bx);

    \draw[postbvb] (Bx) to[] node[anchor=south] {$s_3'$} (B0x);
    \draw[postbvb] (Bx) to[] node[anchor=north] {$s_4'$} (B1x);

    \draw[post] (B0x) to[] node[anchor=north] {$s_5'$} (D0x);
    \draw[postbvb] (B0x) to[] node[anchor=south] {$s_6'$} (B01x);
    \draw[postbvb] (B1x) to[] node[anchor=north] {$s_7'$} (B01x);
    \draw[post] (B1x) to[] node[anchor=north] {$s_8'$} (E1x);

    \draw[post] (B01x) to[] node[anchor=south] {$s_9'$} (D0x);
    \draw[post] (B01x) to[] node[anchor=south] {$s_{10}'$} (E0x);
    \draw[post] (B01x) to[] node[anchor=north] {$s_{11}'$} (E1x);

    %
    \draw[rule,loop] (B) to[] (B);
    \draw[rule,loop] (B0) to[] (B0);
    \draw[rule,loop] (B1) to[] (B1);
    \draw[rule,loop] (B01) to[min distance=5mm,in=65,out=98,looseness=2] (B01);

    \draw[rule,loop] (Bx) to[] (Bx);
    \draw[rule,loop] (B0x) to[min distance=5mm,in=135,out=178,looseness=2] (B0x);
    \draw[rule,loop] (B1x) to[] (B1x);
    \draw[rule,loop] (B01x) to[] (B01x);



    \draw[postdash] (D0x.north) -- ($(D0x)+(0,.5)$) -| ($(V0.north west)+(-0.45,0)$) -- (V0.north west);
    \draw[postdash] (E0x.east) -|($(D0x)+(0.5,0.75)$) -| ($(V0.west)+(-0.6,0)$) -- (V0.west);

    \draw[postdash] (E1x.east) -|($(E1x)+(0.5,-1.55)$) -| ($(V1.west)+(-0.6,0)$) -- (V1.west);

\end{tikzpicture}}
        \caption{The simplified threshold automaton of the Byzantine consensus of Algorithm~\ref{alg:bbc-safe-cc} obtained after model checking the bv-broadcast. Rules $s_j'$, $1\le i \le 11$, are obtained from $s_j$ by replacing each variable~$c\in\{a_0, a_1, \textit{bvb}_0, \textit{bvb}_1\}$ with its corresponding one~$c'$.
        }
        \label{fig:dbft_safe_comp}
    \end{center}
\end{figure*}

A superround $R$ of the simplified automaton from
Fig.~\ref{fig:dbft_safe_comp} captures round~$2R{-}1$ followed by
round~$2R$ of Algorithm~\ref{alg:bbc-safe-cc}.
One can thus restate Lemma~\ref{lem:sameest}
as the following corollary in
the TA terminology. The proof is deferred to Appendix~\ref{app:corollary}.
\begin{corollary}
    \label{coro:consequence-of-good-bvb}
    Let $r \in \nats$ be such that the $r^{\ms{th}}$ execution of
    $\BVbcast$ in Algorithm~\ref{alg:bbc-safe-cc} is
    $(r \modulo 2)$-good.
    Then:
    \begin{itemize}[noitemsep,topsep=0pt]
        \item If there exists $R \in \nats$ with $r = 2R{-}1$, then
              $\ltlG \bigl(\kappa[M_0,R]=0)$ holds.
        \item If there exists $R \in \nats$ with $r = 2R$, then
              $\ltlG(\kappa[M'_1,R]=0\bigr)$ holds.
    \end{itemize}
\end{corollary}

\section{Verification of Byzantine Consensus}\label{sec:verification}

In this section we formally prove that Algorithm~\ref{alg:bbc-safe-cc} solves the Byzantine consensus problem
with the fair \BVbcast\ and without partial
synchrony. (Appendix~\ref{app:no-termination} provides a
counter-example illustrating why the algorithm does not terminate
without the fair broadcast.)  In particular, we apply a methodology
developed for crash fault tolerant randomized
consensus~\cite{BertrandKLW19} to our context to prove both the safety
(Section~\ref{ssec:safety}) and liveness (Section~\ref{ssec:liveness})
properties of the deterministic Byzantine consensus algorithm.

\subsection{Safety}\label{ssec:safety}
Under no fairness assumption, one can prove the safety properties---agreement and validity---of the Byzantine consensus based
on \BVbcast.
Precisely, we formulate these properties in \LTL\ and want to establish that they hold
on the threshold automaton of Fig.~\ref{fig:dbft_safe_comp}.

Agreement requires that no two correct processes disagree, that is, if one process decides~$v$ then no process should decide $1{-}v$ for all binary values $v\in\{0,1\}$.
Thus, we want to prove that the following formula holds for both values of~$v$:
\begin{equation*}
    \label{eq:agree-cc}
    \tag{$\ms{Agree}_v$}
    \forall R\in\nats, \forall R'\in\nats\;
    \Bigl(\ltlF  \kappa[D_v,R] \neq 0  \;\Rightarrow\; \ltlG \kappa[D_{1-v},R']=0 \Bigr) \enspace,
\end{equation*}
stating that for any two superrounds~$R$ and~$R'$, if eventually a process decides~$v$, then globally (in any superround) no process will decide $1-v$. In terms of the TA from Fig.~\ref{fig:dbft_safe_comp}, if a process enters location~$D_v$ no process should enter location $D_{1-v}$ (not only in that superround, but in any other).

Validity requires that if no process proposes a value~$v\in\{0,1\}$, no process should ever decide that value.
Hence, we want to prove the following formula for both values of~$v$:
\begin{equation*}
    \label{eq:valid-cc}
    \tag{$\ms{Valid}_v$}
    \forall R\in\nats\; \Bigl(
    \kappa[V_v,1] =0 \;\Rightarrow\;
    \ltlG \kappa[D_v,R] =0 \Bigr) \enspace,
\end{equation*}
stating that if initially no process has value $v$, then globally (in any superround) no process  decides~$v$.
In terms of the TA, if location $V_v$ is initially empty (in superround~1), then no process should enter location~$D_v$ in any superround.

ByMC can only check formulas of the form
$\forall R\in\nats\; \varphi[R]$ (see
Appendix~\ref{app:multi2oneround}). Thus, automatically checking
\eqref{eq:agree-cc} and \eqref{eq:valid-cc} is non-trivial, as
they both involve two superround numbers: $R$ and $R'$ in
\eqref{eq:agree-cc}, and  1 and $R$ in
\eqref{eq:valid-cc}.
We instead check
well-chosen one-superround
invariants~\eqref{eq:inv1_v} and~\eqref{eq:inv2_v}:

\begin{equation*}
    \label{eq:inv1_v}
    \tag{$\ms{Inv1}_v$}
    \forall R\in\nats\; \Bigl(\ltlF \kappa[D_v,R] \neq 0
    \;\Rightarrow \;
    \ltlG \bigl(\kappa[D_{1-v},R] = 0 \wedge \kappa[E'_{1-v},R] = 0 \bigr)\Bigr) \enspace,
\end{equation*}
\begin{equation*}
    \label{eq:inv2_v}
    \tag{$\ms{Inv2}_v$}
    \forall R\in\nats\; \Bigl(\ltlG \kappa[V_v,R] = 0
    \;\Rightarrow \;
    \ltlG \bigl(\kappa[D_v,R] = 0 \wedge \kappa[E'_v,R] = 0  \bigr)\Bigr)\enspace.
\end{equation*}

The choice of these invariants follows a previous approach used for
the crash fault tolerant
consensus 
where it is shown
that these invariants imply (\ref{eq:agree-cc}) and
(\ref{eq:valid-cc})~\cite[Proposition~2]{BertrandKLW19}. Intuitively,
this follows from the fact that (i) emptiness of $D_0$ and $E_0'$ in
one superround leads to the emptiness of $V_0$ in the next superround,
and (ii) emptiness of $E_1'$ (and $D_1$) in one superround leads to
the emptiness of $V_1$ in the next superround.  Therefore, in order to
prove agreement and validity, we only need to prove (\ref{eq:inv1_v})
and~(\ref{eq:inv2_v}) for both values $v\in\{0,1\}$. We successfully
do this automatically with ByMC (see Section~\ref{sec:expe}).


\subsection{Liveness}\label{ssec:liveness}
We now aim at proving termination of Algorithm~\ref{alg:bbc-safe-cc}.
First, we need to prove that every superround eventually terminates,
in the sense that for every round eventually there are no processes in
any location of that round to the exception of the final ones ($D_0$,
$E_0'$ and $E_1'$).  Formally, using ByMC we prove the following:
\begin{equation*}
    \label{eq:roundterm}
    \tag{$\ms{SRound Term}$}
    \forall R\in\nats\; \ltlF \Bigl( \bigwedge_{L\in \mathcal{L}\setminus\{D_0,E_0',E_1'\}} \kappa[L,R]=0\Bigr)\enspace.
\end{equation*}
From this property and the shape of the TA from  Fig.~\ref{fig:dbft_safe_comp},
it easily follows that if no process ever enters $E_0'$ and $E_1'$ of some superround, then
all processes visit~$D_0$ in that superround.
Similarly, if no process ever enters $E_0$ and $E_1$ of some superround, then
all processes visit~$D_1$ in that superround.
This allows us to express termination as the following \LTL\ property on the threshold
automaton of Fig.~\ref{fig:dbft_safe_comp}:
\begin{equation*}
    \label{eq:term}
    \tag{$\ms{Term}$}
    \exists R\in\nats\; \Bigl( \ltlG \bigl(\kappa[E_0,R] =0 \wedge \kappa[E_1,R] =0 \bigr)
    \;\vee\;
    \ltlG\bigl(\kappa[E'_0,R] =0 \wedge \kappa[E'_1,R] =0 \bigr) \Bigr)\enspace.
\end{equation*}
In words, there is a superround $R$ in which either (i) all processes  visit $D_1$,
or (ii) all processes  visit $D_0$.
%
Here again formula~\eqref{eq:term} is non-trivial to check since it
contains an existential quantifier over superrounds, that cannot be
handled by the model checker ByMC.  Adapting the technique
from~\cite[Section~7]{BertrandKLW19} to a non-randomized context, it
is sufficient to prove a couple of properties on the threshold
automaton of Fig.~\ref{fig:dbft_safe_comp}, that we detail below. The
first property expresses that if no process starts a superround $R$
with value $v$, then all processes decide $1{-}v$ in superround $R$:
\begin{equation*}
    \label{eq:c2}
    \tag{$\ms{Dec}$}
    \begin{split}
        \forall R\in\nats\; &
        \Bigl(\ltlG \bigl( \kappa[V_0,R] =0\bigr) \Rightarrow
        \ltlG \bigl(\kappa[E_0,R] =0 \wedge \kappa[E_1,R] =0 \bigr) \Bigr) \\
        \wedge\;\;	&
        \Bigl(\ltlG \bigl( \kappa[V_1,R] =0\bigr) \Rightarrow
        \ltlG \bigl(\kappa[E_0',R] =0 \wedge  \kappa[E_1',R]= 0 \bigr) \Bigr) \enspace.
    \end{split}
\end{equation*}
The second property claims that (i) emptiness of $M_0$ in superround~$R$ implies
(emptiness of $E_0$ and therefore also) emptiness of $D_0$ and $E_0'$ in~$R$ and
(ii) emptiness of $M_1'$ in superround~$R$ implies emptiness of $E_1'$ in~$R$:
    \begin{multline}
    \label{eq:lem41}
    \tag{$\ms{Good}$}
    \forall R\in\nats\;
    \Bigl(  \bigl( \ltlG \kappa[M_0,R]=0)
    \;\Rightarrow\;
    \ltlG (\kappa[D_0,R] \wedge \kappa[E_0',R]=0) \bigr)\\
    \;\wedge\; 
    \bigl( \ltlG \kappa[M_1',R]=0)
    \;\Rightarrow\;
    \ltlG \kappa[E'_1,R]=0 \bigr)\Bigr) \enspace.
   \end{multline}

The main idea is to exploit the fairness of {\BVbcast},
which ensures the existence of a round $r$ which is $(r \modulo 2)$-good.
Intuitively, the next superround $R = \lceil r/2\rceil$ is the desired witness
for~\eqref{eq:term}, namely the one in which all processes decide (not necessarily for the first time).
We formalize this in our main result:
\begin{theorem}
    Assuming fairness of the \BVbcast, Algorithm~\ref{alg:bbc-safe-cc} terminates.
\end{theorem}

\begin{proof}
    First we prove formulas \eqref{eq:roundterm} and \eqref{eq:c2} and
    \eqref{eq:lem41} automatically using the model checker ByMC.
    Formula~\eqref{eq:roundterm} guarantees that formula~\eqref{eq:term}
    indeed expresses termination. Next, we show that formulas \eqref{eq:c2} and
    \eqref{eq:lem41} together imply \eqref{eq:term}.
    Indeed, since we assume fairness of the
        {\BVbcast}, from Corollary~\ref{coro:consequence-of-good-bvb} we
    know that there is a superround~$R$ in which one of the following
    two scenarios happen:
    \begin{itemize}[noitemsep,topsep=0pt]
        \item $\ltlG \kappa[M_1',R]=0$. In this case
              formula~(\ref{eq:lem41}) implies $\ltlG \kappa[E'_1,R]=0$.
              Note that the form of the (dotted) round-switch rules yield
              that no process starts the superround $R{+}1$ with value~$1$,
              that is, we have $\ltlG \kappa[V_1,R{+}1]=0$.
              Then formula~\eqref{eq:c2} implies
              $\ltlG \bigl(\kappa[E_0',R{+}1] =0 \wedge  \kappa[E_1',R{+}1]= 0 \bigr)$,
              which makes formula~\eqref{eq:term} true, that is,
              all processes visit $D_0$ in superround $R{+}1$. 
        \item $\ltlG \kappa[M_0,R]=0$. In this case
              formula~\eqref{eq:lem41} implies
              $\ltlG (\kappa[D_0,R] \wedge \kappa[E_0',R]=0) \bigr)$.
              Now the round-switch rules yield that no process starts the
              superround~$R{+}1$ with value~$0$, that is, we have
              $\ltlG \kappa[V_0,R{+}1]=0$. Then formula~\eqref{eq:c2}
              implies $\ltlG \bigl(\kappa[E_0,R{+}1] =0 \wedge \kappa[E_1,R{+}1] =0 \bigr)$,
              which satisfies formula~\eqref{eq:term}, that is,
              all processes visit $D_1$ in $R{+}1$.
    \end{itemize}
    As a consequence, our automated proofs of properties
    \eqref{eq:roundterm} and \eqref{eq:c2} and
    \eqref{eq:lem41} guarantee termination of
    Algorithm~\ref{alg:bbc-safe-cc} under fairness of
    \BVbcast.
\end{proof}
\section{Experiments}\label{sec:expe}
In this section, we model check the safety but also the liveness
properties of Byzantine consensus for any parameters $t$ and $n>3t$.
In particular, we show that we formally verify the simplified representation of the
blockchain consensus
in less than 70\,seconds, whereas we could not model check its naive representation.

\myparagraph{Experimental settings}
We used the parallelized version of ByMC 2.4.4 with MPI. The $\BVbcast$ and the simplified automaton were verified on a laptop with Intel® Core™ i7-1065G7 CPU @ 1.30GHz × 8 and 32 GB of memory.  The naive Threshold Automaton (\TA) timed-out even on  a  4  AMD Opteron  6276  16-core  CPU  with  64  cores  at  2300MHz with 64 GB of memory. $\ms{Good}$ and $\ms{Dec}$ are only relevant for the simplified automaton.   The specification of the termination for ByMC is deferred to Appendix~\ref{sec:ta-spec-termination}.

\begin{table}[htbp]
    \setlength{\tabcolsep}{4pt}
    \begin{center}
        \begin{tabular}{ r|l|llll}
            \toprule
            \textbf{TA}                                                            & \textbf{Size}                      & \textbf{Property}                           & \textbf{\# schemas} & \textbf{Avg. length} & \textbf{Time} \\
            \hline
            \multirow{4}{6em}{$\BVbcast$ (Fig.~\ref{fig:bvb-ta})}             & \multirow{4}{6em}{4 unique guards                                                                                                             \\10 locations\\19 rules}  & $\ms{BV-Just_0}$                            & 90                  & 54                          & 5.61s              \\
                                                                                   &                                    & $\ms{BV-Obl_0}$                             & 90                  & 79                   & 6.87s         \\
                                                                                   &                                    & $\ms{BV-Unif_0}$                            & 760                 & 97                   & 27.64s        \\
                                                                                   &                                    & $\ms{BV-Term}$                              & 90                  & 79                   & 6.75s         \\
            \midrule
            \multirow{4}{6em}{Naive consensus (Fig.~\ref{fig:large-ta})}           & \multirow{4}{6em}{14 unique guards                                                                                                            \\24 locations\\45 rules}                         & $\ms{Inv1_0}$                               & >100 000            & -                           & >24h               \\
                                                                                   &                                    & $\ms{Inv2_0}$                               & >100 000            & -                    & >24h          \\
                                                                                   &                                    & $\ms{SRound-Term}$                          & >100 000            & -                    & >24h          \\
                                                                                   &                                    &                                             &                     &                      &               \\

            \midrule
            \multirow{4}{6em}{Simplified consensus (Fig.~\ref{fig:dbft_safe_comp})} & \multirow{4}{6em}{10 unique guards                                                                                                            \\16 locations\\37 rules} & $\ms{Inv1_0}                              $ & 6                   & 102                         & 4.68s              \\
                                                                                   &                                    & $\ms{Inv2_0}                              $ & 2                   & 73                   & 4.56s         \\
                                                                                   &                                    & $\ms{SRound-Term}$                          & 2                   & 109                  & 4.13s         \\
                                                                                   &                                    & $\ms{Good_0}$                               & 2                   & 67                   & 4.55s         \\
                                                                                   &                                    & $\ms{Dec_0}$                                & 2                   & 73                   & 4.62s         \\

            \bottomrule
        \end{tabular}
        \vspace{1em}
        \caption{Although none of the properties of the naive blockchain consensus could be verified within a day of execution of the model checker, it takes about $\mathtt{\sim}4$\,s to verify each property on the simplified representation of the blockchain consensus. Overall it takes less than 70\,seconds to verify both that the binary value broadcast and the simplified representation of the blockchain consensus are correct.
            \label{tab:experiments}}
    \end{center}
\end{table}

\myparagraph{Results}
Table~\ref{tab:experiments} depicts the time (6th column) it takes to verify each property (3rd column) automatically. In particular, it lists the
\TA (1st column) on which these properties were tested, as well as the size of these \TA (2nd column) as the number of guards locations and rules they contain.  A schema (4th column) is a sequence of unlocked guards (contexts) and rule sequences that is used to generate execution paths~\cite{KLVW17popl} whose average length appears in the 5th column.
It
demonstrates the efficiency of our
approach as it allows to verify all properties of the Byzantine
consensus automatically in less than 70 seconds whereas a
non-compositional approach timed out. Although not indicated here, we also generated a counter-example
of $\ms{Inv1_0}$ for $n > 3t$ on the composite automaton in $\mathtt{\sim}4$\,s.

%

\section{Related Work}
\label{sec:rel}


Interactive theorem provers~\cite{SergeyWT18,RahliGBC15,GleissenthallKB19} already checked proofs of consensus algorithms used in the blockchain industry.
In particular, Coq helped prove 
the Raft consensus algorithm~\cite{WilcoxWPTWEA15}, which is not Byzantine fault tolerant but part of crash fault tolerant distributed ledgers~\cite{BCG16,ABB18}, and the Byzantine consensus algorithm of the Algorand blockchain~\cite{ACLMPPR19}.
In addition, Dafny~\cite{HawblitzelHKLPR15} proved MultiPaxos, a consensus algorithm that tolerates crash failures.
Isabelle/HOL~\cite{NPW02} was used to prove Byzantine fault tolerant algorithms~\cite{CBDM11} and was combined with Ivy to try to prove the Byzantine consensus protocol of the Stellar blockchain~\cite{LD20} as discussed in the introduction.
Theorem provers check proofs, not the algorithms.
Hence, one has to invest efforts into writing detailed mechanical proofs.

Specialized decision procedures are a way of proving consensus algorithms.
They were used to prove Paxos~\cite{KraglEHMQ20}, which could itself be used in the aforementioned crash fault tolerant distributed ledgers.
Crash fault tolerant consensus algorithms were manually encoded with their invariants and properties
to prove formulae using the Z3 SMT solver~\cite{DragoiHVWZ14}.
Decision procedures also proved the safety of Byzantine fault tolerant consensus algorithms when $f=t$~\cite{BerkovitsLLPS19}
but not their termination.
Similarly, a proof by refinement of the safety of a Byzantine variant of Paxos was proposed~\cite{Lam11a} but its liveness is not proven.
These decision procedures require the
user to fit the specification into a suitable logical fragment.

Explicit-state model checking fully automates verification of distributed algorithms~\cite{H2003,YPL99}.
It allows to check the reliable broadcast algorithm~\cite{JKSVW13:SPIN}, a common component of various blockchain consensus algorithms~\cite{MXC16,CGLR18,CGG19}.
TLC~\cite{YPL99} checked a reduction of fault tolerant distributed algorithms in the Heard-Of model
that exploits their communication-closed property~\cite{Chaouch-SaadCM09}.
And the agreement of consensus algorithms was proved in the asynchronous setting~\cite{NoguchiTK12}.
These explicit-state tools enumerate all reachable states and thus suffer from state explosion.

Symbolic model checkers~\cite{BCMDH90} cope with this explosion by
representing state transitions efficiently.
NuSMV and SAT helped check consensus algorithms for up
to 10 processes~\cite{TS08,TsuchiyaS11}.
Apalache~\cite{KT19}  uses satisfiability modulo theories (SMT) to check inductive
invariants and verify symbolic executions of~\tlap{}
specifications of the reliable broadcast and crash fault tolerant consensus algorithms but requires
parameters to be fixed.
These tools cannot be used to prove (or disprove) correctness for an arbitrary number of
processes.

Parameterized model checking~\cite{DF99} works for an
arbitrary number $n$ of processes~\cite{2015Bloem}.
Although the problem is undecidable~\cite{AK86} in general, one
can verify specific classes of algorithms~\cite{EK00}.
Indeed, distributed algorithms with a ring-based topology were checked with automata-theoretic method~\cite{AiswaryaBG18} and with Presburger arithmetics formulae verified by an SMT solver~\cite{SangnierSPT20}.
Bosco~\cite{SR08} has been the focus of various parameterized verification techniques~\cite{LWB17,Balasubramanian20}, however,
it acts as a fast path wrapper around
a separate correct consensus algorithm that remains itself to be proven.
The condition-based consensus algorithm~\cite{MRR03,MMPR03} was verified~\cite{Balasubramanian20} with the Byzantine model checker ByMC~\cite{KLVW17popl,KW18,marilinajogor}, only under the condition that the difference between the numbers of processes initialized with 0 and 1 differ by at least $t$.
Recently, the crash fault tolerant Ben-Or consensus algorithm was proved correct with a probabilistic reasoning extension of ByMC~\cite{BertrandKLW19}.
In this paper, we also exploit ByMC but prove the Byzantine consensus algorithm~\cite{CGLR18} of an existing blockchain~\cite{CNG21}.

\section{Conclusion}\label{sec:conclusion}

We presented the first formal verification of a blockchain consensus
algorithm thanks to a new holistic approach. 
Previous attempts to formally verify the liveness of blockchain consensus consisted of
verifying different parts of the consensus algorithm without verifying the sum of the parts.
By modeling directly the pseudocode into a disambiguated threshold automaton 
we guarantee that the ``actual''  algorithm is verified.
By model checking the threshold automaton without the need for user-defined invariants
and proofs, we drastically reduce the risks of human errors.  
We believe that this holistic verification technique will help verify or identify bugs  
in other distributed algorithms based on various broadcast primitives.  


\bibliographystyle{ACM-Reference-Format}
\bibliography{references}

\appendix

\section{Reducing multi-round \TA to one-round \TA}\label{app:multi2oneround}
Let us first formally define a (finite or infinite) \emph{run} in a (one-round or multi-round) counter system~$\ms{Sys}(\TA)$. It is an alternating sequence of configurations and transitions $\sigma_0, t_1,\sigma_1, t_2,\ldots$ such that $\sigma_0\in I$ is an initial configuration and for every $i\ge 1$ we have that $t_i$ is unlocked in $\sigma_{i-1}$, and executing it leads to~$\sigma_i$, denoted $t_i(\sigma_{i-1})=\sigma_i$.

Here we briefly describe the reasoning behind the reduction of multi-round \TA{}s to one-round TAs~\cite[Theorem~6]{BertrandKLW19}.
Note that the behavior of a process in one round only depends on the variables (the number of messages) of that round.
Namely, we check if a transition is unlocked in a round by evaluating a guard and a location counter in that round.
This allows us to modify a run by swapping two transitions from different rounds, as they do not affect each other, and preserve \LTLX properties, which are properties expressed in \LTL{} without the next operator~$\ltlX$.
The type of swapping we are interested in is the one where a transition of round~$R$ is followed by a transition of round~$R'<R$.
Starting from any (fully asynchronous) run, if we keep swapping all such pairs of transitions, we will obtain a run in which processes synchronize at the end of each round and which has the same \LTLX properties as the initial one.
This, so-called \emph{round-rigid} structure, allows us to isolate a single round and analyze it.
Still, different rounds might behave differently as they have different initial configurations.
If we have a formula $\forall R\in\nats.\; \varphi[R]$, where $\varphi[R]$ is in the above mentioned fragment of (multi-round) \LTL, then Theorem~6 of~\cite{BertrandKLW19} shows exactly that it is equivalent to check that (i) this formula holds (or $\varphi[R]$ holds on all rounds~$R$) on a multi-round $\TA$, and (ii) formula $\varphi[1]$ (or just $\varphi$) holds on the one-round $\TA'$ (naturally obtained from the $\TA$ by removing dotted round-switch rules) with respect to all possible initial configurations of all rounds.
Thus, we can verify properties of the form $\forall R\in\nats.\; \varphi[R]$ on multi-round threshold automata, by using ByMC to check $\varphi$ on a one-round threshold automaton with an enlarged set of initial configurations.

\section{Examples of fairness and of non-termination without fairness}\label{app:no-termination}

First, we explain that the fairness is satisfied as soon as one execution of \BVbcast\ has correct processes delivering all values broadcast by correct processes first.
Then, we explain that the Byzantine consensus algorithm cannot terminate without an additional assumption, like fairness.

\paragraph*{Relevance of the fairness assumption}
It is interesting to note that our fairness assumption is satisfied by the existence of an execution with a particular reception order of some messages of the two broadcasts within the \BVbcast. Consider that $t = \lceil n/3 \rceil -1$ and that at the beginning of a round $r$, the two following properties hold:
(i) estimate $r \modulo 2$ is more represented than estimate $(1-r) \modulo 2$ among correct processes and (ii) all correct processes deliver the values broadcast
by correct processes before any value broadcast during the \BVbcast\  by Byzantine processes are delivered. Indeed, the existence of such a round $r$ in any infinite sequence of
executions of \BVbcast\ implies that this sequence is fair (Def.~\ref{def:fair-bvb}): as
$r \modulo 2$ is the only value that can be broadcast by $t{+}1$ correct processes, this is the first value that is received from $t{+}1$ distinct processes
and rebroadcast by the rest of the correct processes. This is thus also the first value that is bv-delivered by all correct processes.

\paragraph*{Non-termination without fairness}
It is interesting to note why Algorithm~\ref{alg:bbc-safe-cc} does not
solve consensus when $t<n/3$ and without our \fairness assumption.  We
exhibit an example of execution of the algorithm with $n=4$ and $f=1$,
starting at round $r$ and for which the estimates of the correct
processes are kept as
$(1 - r) \modulo 2, (1 - r)  \modulo 2, r \modulo 2$ in rounds $r$ and
$r{+}2$. Repeating this while incrementing $r$ yields an infinite
execution, so that the algorithm never terminates.

\begin{lemma}
    Algorithm~\ref{alg:bbc-safe-cc} does not terminate without fairness.
\end{lemma}
\begin{proof}
    Consider, for example, processes $p_1, p_2, p_3$ and $p_4$ where $p_4$ is Byzantine and where
    $0, 0, 1$ are the input values of the correct processes $p_1, p_2, p_3$, respectively, at round 1.
    We show that at the beginning of round 2, $p_1, p_2, p_3$ have estimates $0, 1, 1$.
    First, as a result of the broadcast (line~\ref{line:bvb-first-bcast}), consider that
    $p_1$ and $p_2$ receive 0 from $p_1$, $p_2$ and $p_4$ so that $p_1, p_2$ bv-deliver $0$.
    Second, $p_2$ and $p_3$ receive 1 from $p_3$,  $p_4$ and finally $p_2$ so that $p_2, p_3$ bv-deliver $1$.
    Third, $p_3$ receives 0 from $p_0$, $p_2$ and finally from $p_3$ itself, hence $p_3$ bv-delivers $0$.
    Now we have:
    (a) $p_1, p_2, p_3$ bv-deliver $0$, $0$, $1$ and
    (b) $p_2, p_3$ later bv-deliver $1$ and $0$, respectively.
    As a result of (a), we have
    $p_1, p_2$ broadcast, and say $p_4$ sends, $\tup{\AUX, 0, \cdot}$ so that $p_0$ receives these three messages,
    $p_1, p_2$ broadcast $\tup{\AUX, 0, \cdot}$, and say $p_4$ sends, $\tup{\AUX, 1, \cdot}$ to $p_2$ so that $p_2$ receives these messages,
    $p_1$ broadcasts $\tup{\AUX, 0, \cdot}$ while $p_3$ broadcasts, and say $p_4$ sends, $\tup{\AUX, 1, \cdot}$ so that $p_3$ receives these messages.
    Finally, by (b) we have $\ms{contestants}_2 = \ms{contestants}_3 = \{0, 1\}$.
    This implies that the $n-t$ first values inserted in $\ms{favorites}_1$, $\ms{favorites}_2$ and $\ms{favorites}_3$ in round $r$ are values $\{0\}$, $\{0,1\}$, $\{0,1\}$, respectively.
    Finally, $\ms{qualifiers}_1$, $\ms{qualifiers}_2$ and $\ms{qualifiers}_3$ are $\{0\}$, $\{0,1\}$ and $\{0,1\}$, respectively.
    And $p_1, p_2, p_3$ set their estimate to $0, 1, 1$.

    It is easy to see that a symmetric execution in round $r'=r+1$ leads processes to change their estimate from $0, 1, 1$
    to $0, 0, 1$ looping back to the state where $r \modulo 2 = 1$ and estimate are $(1 - r) \modulo 2, (1 - r) \modulo 2, r \modulo 2$.
\end{proof}

\section{Starting a round with identical estimate}\label{app:same-est}

\begin{lemma}[Lemma~\ref{lem:sameest}]
    If the infinite sequence of bv-broadcast invocations of Algorithm~\ref{alg:bbc-safe-cc} is fair,
    with the $r^{\ms{th}}$ invocation (in round~$r$) being $(r \modulo 2)$-good,
    then all correct processes start round~$r{+}1$ of
    Algorithm~\ref{alg:bbc-safe-cc} with estimate $r \modulo 2$.
\end{lemma}
\begin{proof}
    The argument is that all correct processes wait until a growing prefix of the bv-delivered values that are re-broadcast implies that there is a subset of favorites, called $\ms{qualifiers}$, containing messages from $n-t$ distinct processes such that $\forall v \in \ms{qualifiers}.\; v \in \ms{contestants}$.
    As we assume that the infinite sequence of bv-broadcast invocations of Algorithm~\ref{alg:bbc-safe-cc} is fair,
    with the $r^{\ms{th}}$ invocation being $(r \modulo 2)$-good, then we know that in round~$r$ for every pair of correct processes $p_i$ and $p_j$ we have $\ms{p_i.qualifiers} \subseteq \ms{p_j.qualifiers}$ or $\ms{p_j.qualifiers} \subseteq \ms{p_i.qualifiers}$.
    If $\ms{p_i.qualifiers} = \ms{p_j.qualifiers}$ for all pairs, then by examination of the code, we know that they will set their estimate $\ms{est}$ to the same value depending on the parity of the current round.

    Consider instead, with no loss of generality, that
    $\ms{p_i.qualifiers}$ is a strict subset of $\ms{p_j.qualifiers}$
    in round $r$. As their values can only be binaries, in $\{0,1\}$,
    this means that $\ms{p_i.qualifiers}$ is a singleton, say $\{w\}$.
    As all correct processes bv-deliver $r \modulo 2$ first, which is
    then broadcast into $\ms{p_i.favorites}$, we have
    $w = r \modulo 2$ and $p_i$'s estimate becomes $r \modulo 2$ at
    line~\ref{BYZ-safe-11}.  As $\ms{p_j.qualifiers}$ is $\{0,1\}$,
    the estimate of $p_j$ is also set to $r \modulo 2$ but at
    line~\ref{BYZ-safe-12}.
\end{proof}

\section{Large TA}\label{app:largeta}

Table~\ref{tab:pta-rules} details the rules for the first half of the threshold automaton from Fig.~\ref{fig:large-ta}.
\begin{table*}[htbp]
    \begin{center}
        \begin{tabular}{lll}
            \text{Rules}             & \text{Guard}             & \text{Update} \\
            \hline
            $r_1$                    & $\mathit{true}$          & $\cpp{b_0}$   \\
            $r_2$                    & $\mathit{true}$          & $\cpp{b_1}$   \\
            $r_3$                    & $b_0\ge 2t{+}1{-}f$      & $\cpp{a_0}$   \\
            $r_4$                    & $b_1\ge t{+}1{-}f$       & $\cpp{b_1}$   \\
            $r_5$                    & $b_0\ge t{+}1{-}f$       & $\cpp{b_0}$   \\
            $r_6$                    & $b_1\ge 2t{+}1{-}f$      & $\cpp{a_1}$   \\
            $r_{14}, r_{15}, r_{16}$ & $a_0\ge n{-}t{-}f$       & ---           \\
            $r_{8}$                  & $b_1 \ge t{+}1{-}f$      & $\cpp{b_1}$   \\
            $r_{9}$                  & $b_1 \ge 2t{+}1{-}f$     & $\cpp{a_1}$   \\
            $r_{10}$                 & $b_0 \ge 2t{+}1{-}f$     & $\cpp{a_0}$   \\
            $r_{11} $                & $b_0 \ge t{+}1{-}f$      & $\cpp{b_0}$   \\
            $r_{12}$                 & $b_1 \ge 2t{+}1{-}f$     & ---           \\
            $r_{13}$                 & $b_0 \ge 2t{+}1{-}f$     & ---           \\
            $r_{7}, r_{18}, r_{19} $ & $a_1\ge n{-}t{-}f$       & ---           \\
            $r_{16}$                 & $a_0\ge n{-}t{-}f$       & ---           \\
            $r_{17}$                 & $a_0{+}a_1\ge n{-}t{-}f$ & ---           \\
            $r_{20}, r_{21}, r_{22}$ & $\mathit{true}$          & ---
        \end{tabular}
    \end{center}
    \caption{The rules of the threshold automaton from Fig.~\ref{fig:large-ta}. We omit self loops that have trivial guard $\mathit{true}$ and no update.
    }
    \label{tab:pta-rules}
\end{table*}

\section{Missing proof of Corollary~\ref{coro:consequence-of-good-bvb}}\label{app:corollary}

We restate here Corollary~\ref{coro:consequence-of-good-bvb} and give its proof.
\begin{corollary}
    Let $r \in \nats$ be such that the $r^{\ms{th}}$ execution of
    $\BVbcast$ in Algorithm~\ref{alg:bbc-safe-cc} is
    $(r \modulo 2)$-good.
    Then:
    \begin{itemize}
        \item If there exists $R \in \nats$ with $r = 2R{-}1$, then
              $\ltlG \bigl(\kappa[M_0,R]=0)$ holds.
        \item If there exists $R \in \nats$ with $r = 2R$, then
              $\ltlG(\kappa[M'_1,R]=0\bigr)$ holds.
    \end{itemize}
\end{corollary}
\begin{proof}
    By definition of an $(r \modulo 2)$-good execution, we know that in
    this particular invocation of $\BVbcast$, all correct processes
    bv-deliver $r \modulo 2$ first.  It follows from
    Lemma~\ref{lem:sameest}, that all correct processes start the next
    round with estimate set to $r \modulo 2$.  There are two cases to
    consider depending on the parity of the round: If $r \modulo 2 = 1$,
    then this is the first round of superround $R$, i.e., $r = 2R-1$. As
    a result, $\ltlG \bigl(\kappa[M_{0},R]=0\bigr)$.  If
    $r \modulo 2 = 0$, then this is the second round of superround $R$,
    i.e., $r = 2R$. As a result, $\ltlG \bigl(\kappa[M'_{1},R]=0\bigr)$.
\end{proof}

\section{Specification of the termination property in the simplified
  threshold automaton for consensus algorithm}

\label{sec:ta-spec-termination}

The reliable communication assumption and the properties
guaranteed by the \BVbcast\ are expressed as preconditions for $\ms{s\_round\_termination}$. The
progress conditions work exactly the same as in
\cite{BertrandKLW19}. However, since the shared counters representing the
\BVbcast\ execution do not represent regular messages, we cannot
directly use the reliable communication assumption. Instead, we use the properties of the
\BVbcast\ that we proved in a separate automaton.

In practice, instead of using progress preconditions on the \BVbcast\ counters in
$\ms{s\_round\_termination}$, such as:

\small
\begin{verbatim}
(locM  == 0 || bvb1 < 1) && (locM  == 0 || bvb0 < 1) && 
(locM1 == 0 || bvb0 < 1) && (locM0 == 0 || bvb1 < 1)     
\end{verbatim}
\normalsize

we use the following:

\small
\begin{verbatim}
/* BV-Termination */
(locM    == 0) && 
/* BV-Obligation */
(locM1 == 0 || bvb0 < T + 1) && (locM0 == 0 || bvb1 < T + 1) &&        
/* BV-Uniformity */
(locM1 == 0 || aux0 == 0) && (locM0 == 0 || aux1 == 0) &&
\end{verbatim}
\normalsize

One can note that we do not use BV-Justification as a precondition in
this specification. Instead, the BV-Justification property is baked in
the structure of the simplified threshold automaton (in the guard of
the transition $M \rightarrow M0, M1$).

The complete specification of the termination property follows:

\small
\begin{minipage}{\textwidth}
    \begin{multicols}{2}
        \begin{verbatim}
s_round_termination: 
<>[](
    (locV0   == 0) &&
    (locV1   == 0) &&

    /* BV-Termination */
    (locM    == 0) && 
    /* BV-Obligation */
    (locM1 == 0 || bvb0 < T + 1) && 
    (locM0 == 0 || bvb1 < T + 1) &&        
    /* BV-Uniformity */
    (locM1 == 0 || aux0 == 0) && 
    (locM0 == 0 || aux1 == 0) &&

    /* Business as usual */
    (locM1   == 0 || aux1 < N - T) &&
    (locM0   == 0 || aux0 < N - T) &&
    (locM01  == 0 || aux0 + aux1 < N - T) &&

    (locD1   == 0) &&
    (locE0   == 0) &&
    (locE1   == 0) &&

    /* BV-Termination */
    (locMx    == 0 ) && 
    /* BV-Obligation */
    (locM1x   == 0 || bvb0x < T + 1) &&
    (locM0x   == 0 || bvb1x < T + 1) &&
    /* BV-Uniformity */
    (locM1x == 0 || aux0x == 0) && 
    (locM0x == 0 || aux1x == 0) &&

    (locM1x   == 0 || aux1x < N - T) &&
    (locM0x   == 0 || aux0x < N - T) &&
    (locM01x  == 0 || aux1x < N - T) &&
    (locM01x  == 0 || aux0x < N - T) &&
    (locM01x  == 0 || aux0x + aux1x < N - T)
    ) 
->
<>(       
    locV0   == 0 &&
    locV1   == 0 &&
    locM    == 0 &&
    locM0   == 0 &&
    locM1   == 0 &&
    locM01  == 0 &&
    locE0   == 0 &&
    locE1   == 0 &&
    locD1   == 0 &&
    locMx   == 0 &&
    locM0x  == 0 &&
    locM1x  == 0 &&
    locM01x == 0
);

inv1_0: <>(locD0 != 0) -> [](locD1 == 0 && locE1x == 0);

inv2_0: [](locV0 == 0) -> [](locD0 == 0 && locE0x == 0);

inv1_1: <>(locD1 != 0) -> [](locD0 == 0 && locE0x == 0);

inv2_1: [](locV1 == 0) -> [](locD1 == 0 && locE1x == 0);

dec_0:  [](locV0 == 0) -> [](locE0 == 0 && locE1 == 0);

dec_1:  [](locV1 == 0) -> [](locE0x == 0 && locE1x == 0);

good_0: [](locM0 == 0) -> [](locD0 == 0 && locE0x == 0);

good_1: [](locM1x == 0) -> [](locE1x == 0);
      
\end{verbatim}
    \end{multicols}
\end{minipage}

\normalsize

\end{document}